%% file: NVBC_spinor.tex
\documentclass[11pt]{amsart}

\usepackage[a4paper,top=2cm,bottom=2cm,left=2cm,right=2cm]{geometry}

\usepackage[T1]{fontenc}
\usepackage{amsmath,amssymb,amsthm}
\usepackage{graphicx}
\usepackage{subcaption}
\usepackage{float}
\usepackage{hyperref}

\newcommand{\sech}{\mathop{\rm sech}\nolimits}

\newtheorem{lemma}{Lemma}[section]

\newtheorem{theorem}{Theorem}[section]
\newtheorem{proposition}{Proposition}[section]

\newtheorem{definition}{Definition}[section]

\theoremstyle{definition}

\numberwithin{equation}{section}
\title[Spinor NLS with NZBC]{Adiabatic perturbation theory for the $F=1$ spinor nonlinear Schr\"odinger equation with nonvanishing boundary conditions}
\author{Vassilios M Rothos}
\address{School of Mechanical Engineering, Faculty of Engineering\\
Aristotle University of Thessaloniki\\
Thessaloniki 54124, Greece}
\email{rothos@auth.gr}
\thanks{Corresponding author. Email: rothos@auth.gr}

\thanks{This is the accepted manuscript of an article published in
\emph{Journal of Physics A: Mathematical and Theoretical},
59 (2026), 295701. The Version of Record is available at
\url{https://doi.org/10.1088/1751-8121/ae8490}.}
\thanks{The research was funded by Aristotle University of Thessaloniki (AUTH) Research Council grants number 73191, 73699, 11682.  The work has been co-financed by the European Union (European Social Fund ESF) and Greek national funds through the Operational Program Education and Lifelong Learning of the National Strategic Reference Framework (NSRF) Research Funding Program D. 534 MIS: 379337: THALES and Marie Curie Actions, People, IRSES}



\date{\today}

\dedicatory{}

\begin{document}
\begin{abstract}
We develop a systematic adiabatic perturbation theory for the integrable $F=1$ spinor nonlinear Schr\"odinger equation under nonvanishing boundary conditions, formulated entirely within the framework of the associated Riemann--Hilbert problem. In this setting, localized nonlinear excitations are characterized by discrete spectral data consisting of a complex eigenvalue and an associated polarization vector. For a general class of small perturbations preserving the background, we derive the perturbation-induced evolution of the scattering data directly at the level of the Riemann--Hilbert problem. In the one-soliton sector, this yields a closed finite-dimensional dynamical system governing the slow evolution of the effective soliton parameters, including the spectral variables, the soliton center and phase, the residue amplitude, and the internal polarization state. The latter evolves according to a constrained dynamical equation with no scalar analogue. For localized perturbations, the modulation equations are expressed in explicit integral form in terms of the one-soliton eigenfunctions, providing a fully computable description of the dynamics. In the limit of vanishing boundary conditions, the resulting system reduces to the perturbation theory obtained in \cite{Doktorov2008}.
\end{abstract}

\maketitle
\input{NVBC_spinor_1.tex}
\input{NVBC_spinor_1a.tex}
\input{NVBC_spinor_2.tex}
\input{NVBC_spinor_3.tex}
\input{NVBC_spinor_4.tex}
\input{NVBC_spinor_5.tex}
\input{NVBC_spinor_6.tex}
\input{NVBC_spinor_7.tex}
\input{NVBC_spinor_8.tex}
\bibliographystyle{amsplain}
\bibliography{NVBC_spinor_references}

\end{document}

%% file: NVBC_spinor_1.tex
\section{Introduction}
\label{NVBC_spinor1}

Nonlinear excitations in multicomponent Bose--Einstein condensates have attracted sustained attention from both physical and mathematical perspectives \cite{AblowitzClarkson1992,SulemSulem1999}. In particular, spinor condensates with hyperfine spin $F=1$ provide a natural setting in which nonlinear wave propagation is coupled to internal spin degrees of freedom, leading to a substantially richer dynamics than in the scalar case. Under suitable constraints on the coupling coefficients, the corresponding mean-field equations reduce to an integrable matrix nonlinear Schr\"odinger system, which serves as a fundamental model for spinor solitons and related coherent structures.

The integrable $F=1$ spinor model occupies a distinguished position within multicomponent nonlinear Schr\"odinger systems. It arises as a physically relevant reduction of the three-component Gross--Pitaevskii equations and admits a matrix Lax representation, allowing for a systematic analysis via inverse scattering techniques. In contrast to the scalar case, the associated spectral data consist not only of discrete eigenvalues but also of internal polarization variables reflecting the spinorial structure of the model.

In realistic settings, however, integrability is only approximate. Small deviations from ideal interaction regimes, external potentials, and weak trapping effects naturally lead to perturbations of the integrable dynamics. In scalar and vector nonlinear Schr\"odinger equations, such effects are described by adiabatic perturbation theory, which yields effective evolution equations for the scattering data. In the presence of nonvanishing boundary conditions, this analysis becomes more delicate, since the spectral problem is formulated on a nontrivial background and the inverse problem is posed on a Riemann surface.

In the spinor case, the situation is further complicated by the presence of internal degrees of freedom. In addition to the usual soliton parameters—velocity, width, position, and phase—one must also describe the evolution of polarization variables. A consistent perturbation theory must therefore be compatible both with the analytic structure imposed by nonvanishing boundary conditions and with the reduction constraints of the spinor model.

To place the present work in context, we briefly recall related developments. A natural point of comparison is the inverse scattering theory for multicomponent nonlinear Schr\"odinger equations with nonzero boundary conditions. In particular, the analysis of the three-component defocusing vector nonlinear Schr\"odinger equation in \cite{BiondiniKrausPrinari2016} provides a complete formulation of the spectral problem, including analytic eigenfunctions, a characterization of the discrete spectrum, and soliton solutions on a nonzero background.

The spinor models considered here belong to the same general class but arise as reductions associated with symmetric spaces of BD.I type in the Cartan classification \cite{GerdjikovKostovValchev2009}. These reductions carry an intrinsic internal structure reflecting spin degrees of freedom, leading to additional constraints on both the Lax pair and the scattering data \cite{TsuchidaWadati1998}.

While inverse scattering with nonvanishing boundary conditions has been extensively studied, both for vector systems \cite{BiondiniKrausPrinari2016} and for spinor condensates \cite{DoktorovRothosKivshar2007}, the effect of perturbations on the corresponding spectral data has received comparatively less attention. Perturbation theory for rapidly decaying solitons has been extensively developed using a variety of approaches, including multi-scale analysis, IST-based techniques, eigenfunction expansions, perturbations of conserved quantities, and numerical simulations \cite{AblowitzPrinariTrubatch2004,FaddeevTakhtajan1987,NovikovEtAl1984,Whitham1999,Yang2010}. 

In related multicomponent settings, adiabatic perturbation theory has been applied to vector solitons in the Manakov model \cite{Chernyavskiy2025,PhysRevA.109.023328}, as well as to spinor systems via Riemann--Hilbert methods \cite{Rothos2024}. An integrable perturbation theory for dark-bright solitons of the defocusing Manakov equation was also developed in \cite{MylonasRothosKevrekidisFrantzeskakis2016}.

The present work may be viewed as a continuation of two complementary directions. The first is the integrable spinor theory, where the matrix nonlinear Schr\"odinger structure and its Lax representation were established in earlier works. The second is the Riemann--Hilbert-based perturbation theory for multicomponent systems with nonvanishing boundary conditions, in which the dynamics of nonlinear excitations is derived from the deformation of scattering data. The contribution of the present paper is to unify these directions within the integrable $F=1$ spinor model with nonvanishing boundary conditions.

More precisely, we develop a systematic adiabatic perturbation theory based on the associated Riemann--Hilbert problem. Starting from the perturbation-induced deformation of the scattering data, we derive evolution equations for the discrete eigenvalues and the corresponding residue vectors. In the one-soliton sector, this yields a closed finite-dimensional dynamical system for the effective soliton parameters, including the internal polarization variables. In particular, the spinor structure manifests itself through a nontrivial tangential evolution equation for the polarization vector, a feature without scalar analogue. The Riemann--Hilbert formulation allows for a systematic treatment compatible with the analytic structure imposed by nonvanishing boundary conditions, which is difficult to achieve via direct perturbative methods.

A related perturbative approach for spinor Bose--Einstein condensates was developed in \cite{Doktorov2008}, where bright soliton dynamics was analyzed within a Riemann--Hilbert framework under vanishing boundary conditions. In that setting, certain soliton parameters remain unchanged to leading order, while the phase acquires a perturbative correction. The present work extends this analysis to nonvanishing boundary conditions, where the spectral problem requires a different analytic framework and the resulting modulation dynamics acquires additional structure. In particular, the perturbation theory is formulated intrinsically at the level of the Riemann--Hilbert data and yields a geometric description of the polarization dynamics as a constrained flow.
\subsection*{Main results}
The main results of the present work can be summarized as follows.

\begin{itemize}

\item[(i)] In Section~\ref{NVBC_spinor2} we formulate the direct and inverse spectral problem for the integrable $F=1$ spinor nonlinear Schr\"odinger equation under nonvanishing boundary conditions and construct the associated matrix Riemann--Hilbert problem.

\item[(ii)] In Section~\ref{NVBC_spinor3} we characterize the discrete spectrum and construct the corresponding one-soliton solutions, parametrized by a complex spectral parameter and a polarization vector; see in particular Theorem~\ref{thm_4_1}.

\item[(iii)] In Section~\ref{NVBC_spinor4} we derive the perturbation-induced evolution of the scattering data directly from the perturbed Riemann--Hilbert problem. In the one-pole sector, this yields explicit evolution equations for the discrete eigenvalue and the residue vector; see Theorem~\ref{thm_5_1}.

\item[(iv)] In Section~\ref{NVBC_spinor5} we show that, in the adiabatic regime, the one-soliton dynamics is governed by a closed finite-dimensional system for the effective parameters, including the polarization variables; see Theorem~\ref{thm_6_1}.

\item[(v)] In Sections~\ref{NVBC_spinor6}--\ref{NVBC_spinor7} we obtain explicit integral representations for the modulation coefficients for a class of localized perturbations preserving the background. The resulting modulation system is given in fully explicit form in Theorem~\ref{thm_8_2}.

\item[(vi)] Finally, we show that, in the limit of vanishing background, the modulation equations reduce to the perturbation theory developed in \cite{Doktorov2008}.

\end{itemize}

This result provides, to the best of our knowledge, the first Riemann--Hilbert-based adiabatic perturbation theory for spinor nonlinear Schr\"odinger equations with nonvanishing boundary conditions.

\subsection*{Main theorem}

The principal result of the present work is the following.

\begin{theorem}[Adiabatic modulation of spinor solitons]
\label{thm_main_intro}

Consider the integrable $F=1$ spinor nonlinear Schr\"odinger equation under nonvanishing boundary conditions, subject to a perturbation
\[
iQ_t + \mathcal{N}[Q] = \epsilon R[Q], 
\qquad 0<\epsilon\ll1,
\]
where $R[Q]$ is sufficiently smooth and localized so as to preserve the asymptotic background.

Let the unperturbed solution correspond to a simple discrete eigenvalue $k_1 \in \mathbb{C}$ with associated polarization vector $\mathbf{p}$.

Then, in the adiabatic regime, the perturbed solution remains, to leading order, within the one-soliton manifold, and its dynamics is governed by a closed finite-dimensional system for the parameters
\[
k_1(t)=\xi(t)+i\eta(t), 
\qquad x_c(t), \quad \phi(t), \quad \rho(t), \quad \mathbf{p}(t),
\]
where $\xi$ and $\eta$ denote the real and imaginary parts of the spectral parameter, $x_c$ is the soliton center, $\phi$ is the total phase, $\rho$ is the residue amplitude, and $\mathbf{p}$ is a unit polarization vector, $\mathbf{p}^\dagger \mathbf{p}=1$.

The evolution is given by
\begin{align}
\dot{k}_1(t) &= -\epsilon\, F(t,k_1), \\
\dot{\mathbf{p}}(t) &= \epsilon\Bigl(M(t,k_1) - \mathbf{p}^\dagger M(t,k_1)\mathbf{p}\, I\Bigr)\mathbf{p}(t),
\end{align}
together with compatible evolution equations for $(x_c,\phi,\rho)$, forming a closed system.

Here $F$ and $M$ are determined, respectively, by the singular and regular parts of the perturbation-induced deformation of the associated Riemann--Hilbert problem, and depend only on the current scattering data.

Moreover, for localized perturbations preserving the background, $F$ and $M$ admit explicit integral representations in terms of the one-soliton eigenfunctions, yielding a fully explicit modulation system.

In the limit of vanishing boundary conditions, the above system reduces to the perturbation theory obtained in \cite{Doktorov2008}.

\end{theorem}
To the best of our knowledge, this is the first formulation of adiabatic perturbation theory for spinor NLS systems with nonvanishing boundary conditions entirely at the level of the Riemann–Hilbert problem.

\subsection*{Organization of the paper}
The paper is organized as follows. In Section~\ref{NVBC_spinor2} we introduce the integrable spinor model, its matrix formulation, and the associated Lax pair. In Section~\ref{NVBC_spinor3} we develop the spectral theory under nonvanishing boundary conditions, construct the Jost solutions, analyze their analytic structure, and formulate the corresponding matrix Riemann--Hilbert problem. We then characterize the discrete spectrum and derive the one-soliton solutions.

In Section~\ref{NVBC_spinor4} we consider perturbations of the integrable model and derive the evolution of the scattering data directly at the level of the Riemann--Hilbert problem. In Section~\ref{NVBC_spinor5} we reformulate these results in terms of effective soliton parameters and obtain the corresponding adiabatic modulation system.

In Section~\ref{NVBC_spinor6} we specialize the general theory to localized perturbations and derive explicit modulation equations. In Section~\ref{NVBC_spinor7} we evaluate these expressions in the one-soliton sector, obtaining closed-form formulas for the adiabatic dynamics. The paper concludes with a discussion of the resulting finite-dimensional system and possible extensions.

%% file: NVBC_spinor_1a.tex
\section{Preliminaries: the integrable spinor model and its Lax representation}
\label{NVBC_spinor1a}

In this section we introduce the integrable $F=1$ spinor nonlinear Schr\"odinger model in a form suitable for the subsequent spectral and perturbative analysis. Our aim is to recast the multicomponent Gross--Pitaevskii system into a matrix formulation which makes explicit the underlying algebraic structure and allows for the construction of a Lax representation.

More precisely, we start from the three-component mean-field description of the spinor condensate and pass to an equivalent matrix nonlinear Schr\"odinger equation, which serves as the natural starting point for the inverse scattering and Riemann--Hilbert formulation developed in the sequel.

Let $q(x,t)$ be a matrix-valued function satisfying the evolution equation
\begin{equation}
i q_t + q_{xx} + 2 q q^\dagger q = 0.
\label{eq:MNLS}
\end{equation}
Equation \eqref{eq:MNLS} belongs to the class of completely integrable nonlinear Schr\"odinger systems and admits a Lax representation, which underlies the spectral analysis carried out below. Scalar and vector nonlinear Schr\"odinger equations arise in a variety of physical contexts, including nonlinear optics, water waves, and Bose--Einstein condensation; see, e.g.,~\cite{AblowitzPrinariTrubatch2004,SulemSulem1999}.

We now relate \eqref{eq:MNLS} to the standard three-component spinor formulation. To this end, consider the complex field
\[
\Phi(x,t)=\bigl(\Phi_1(x,t),\Phi_0(x,t),\Phi_{-1}(x,t)\bigr)^T.
\]
Under the integrability constraint on the interaction coefficients, the corresponding Gross--Pitaevskii system can be written in the form
\begin{equation}
\begin{aligned}
i\partial_t \Phi_1 + \partial_{xx}\Phi_1 + 2\Bigl(|\Phi_1|^2 + 2|\Phi_0|^2\Bigr)\Phi_1 + 2\Phi_0^2 \Phi_{-1}^* &= 0, \\
i\partial_t \Phi_0 + \partial_{xx}\Phi_0 + 2\Bigl(|\Phi_1|^2 + |\Phi_0|^2 + |\Phi_{-1}|^2\Bigr)\Phi_0 + 2\Phi_1 \Phi_{-1}\Phi_0^* &= 0, \\
i\partial_t \Phi_{-1} + \partial_{xx}\Phi_{-1} + 2\Bigl(|\Phi_{-1}|^2 + 2|\Phi_0|^2\Bigr)\Phi_{-1} + 2\Phi_0^2 \Phi_1^* &= 0.
\end{aligned}
\label{eq:SpinorSystem}
\end{equation}

The system \eqref{eq:SpinorSystem} admits an equivalent matrix representation of the form \eqref{eq:MNLS}. More precisely, introducing the matrix field
\begin{equation}
q(x,t)=
\begin{pmatrix}
\Phi_0(x,t) & \Phi_1(x,t) \\
\Phi_{-1}(x,t) & -\Phi_0(x,t)
\end{pmatrix},
\label{eq:qmatrix}
\end{equation}
a direct computation shows that \eqref{eq:SpinorSystem} is equivalent to the matrix equation \eqref{eq:MNLS}.

This representation makes explicit the underlying symmetric space structure of BD.I type, where the matrix potential takes values in the off-diagonal component of a $\mathbb{Z}_2$--graded Lie algebra. This algebraic framework encodes the reduction structure and underlies both the Lax representation and the associated Riemann--Hilbert formulation.

Let $k\in\mathbb C$ be the spectral parameter and consider the linear system
\begin{equation}
\Psi_x = U(x,t,k)\Psi, \qquad \Psi_t = V(x,t,k)\Psi,
\label{eq:LaxPairFull}
\end{equation}
where $\Psi(x,t,k)$ is a vector-valued function and the matrices $U$ and $V$ are given by
\begin{equation}
U(x,t,k) = -ikJ + \mathcal Q(x,t),
\label{eq:Ufull}
\end{equation}
\begin{equation}
V(x,t,k) = -2ik^2 J + 2k \mathcal Q(x,t) - i\bigl(J\mathcal Q_x - \mathcal Q_x J\bigr) + 2i\mathcal Q^2 J.
\label{eq:Vfull}
\end{equation}
Here
\begin{equation}
J = \begin{pmatrix}
I & 0 \\
0 & -I
\end{pmatrix},
\label{eq:Jmatrix}
\end{equation}
is the constant grading matrix, while
\begin{equation}
\mathcal Q(x,t)=
\begin{pmatrix}
0 & q(x,t) \\
- q^\dagger(x,t) & 0
\end{pmatrix}
\label{eq:Qhat}
\end{equation}
is the block-off-diagonal matrix potential associated with the spinor field.

The compatibility condition
\begin{equation}
U_t - V_x + [U,V] = 0
\label{eq:ZeroCurvature}
\end{equation}
is equivalent to the matrix equation \eqref{eq:MNLS}. Thus, the nonlinear evolution of the spinor field is encoded in the zero-curvature representation associated with the pair $(U,V)$.

The structure of the matrices \eqref{eq:Ufull}--\eqref{eq:Qhat} reflects the underlying Lie algebraic reduction. In particular, $\mathcal Q$ takes values in the off-diagonal subspace of a $\mathbb Z_2$-graded Lie algebra, while $J$ defines the corresponding grading. This algebraic decomposition will play a central role in the formulation of the reduction constraints on the scattering data and on the associated Riemann--Hilbert problem.

We now turn to the boundary conditions. In the present work we consider nonvanishing boundary conditions of the form
\begin{equation}
\lim_{x\to\pm\infty} q(x,t)=q_0,
\label{eq:NZBCprelim}
\end{equation}
where $q_0$ is a constant matrix. Correspondingly,
\begin{equation}
\lim_{x\to\pm\infty} \mathcal Q(x,t)=\mathcal Q_0
=
\begin{pmatrix}
0 & q_0 \\
- q_0^\dagger & 0
\end{pmatrix}.
\label{eq:Qhat0}
\end{equation}
The asymptotic form of the $x$-part of the Lax operator is therefore
\begin{equation}
U_0(k) = -ikJ + \mathcal Q_0.
\label{eq:U0prelim}
\end{equation}

In contrast to the vanishing-boundary case, the asymptotic operator $U_0(k)$ is not diagonal in the standard basis. Its diagonalization leads to algebraic eigenvalue branches depending nontrivially on the spectral parameter $k$. As a consequence, the spectral problem is naturally formulated on a multi-sheeted Riemann surface on which the eigenvalues and eigenprojections become single-valued analytic functions.

In order to prepare for the direct scattering construction, we rewrite the $x$-part of the Lax pair in canonical scattering form:
\begin{equation}
\Psi_x = \bigl(-ikJ + \mathcal Q(x,t)\bigr)\Psi.
\label{eq:LaxCanonical}
\end{equation}
The direct scattering problem is formulated by comparing solutions of \eqref{eq:LaxCanonical} with those of the constant-coefficient asymptotic system generated by \eqref{eq:U0prelim}.

Finally, the block structure of $\mathcal Q$ implies the reduction symmetry
\begin{equation}
U^\dagger(x,t,k^*) = -U(x,t,k), \qquad
V^\dagger(x,t,k^*) = -V(x,t,k),
\label{eq:LaxReduction}
\end{equation}
which induces corresponding symmetry constraints on the scattering matrix and on the analytic eigenfunctions constructed later.

The direct scattering problem is therefore formulated by comparing the variable-coefficient system \eqref{eq:LaxCanonical} with the constant-coefficient asymptotic problem generated by \eqref{eq:U0prelim}. Since the asymptotic operator is not diagonal in the standard basis, its spectral decomposition leads to a nontrivial Riemann surface structure for the spectral parameter. The associated analytic eigenfunctions and the corresponding matrix Riemann--Hilbert problem will be constructed in the next section.  

The formulation above provides the complete algebraic and analytic framework required for the spectral analysis. The direct scattering problem is formulated by comparing the variable-coefficient system \eqref{eq:LaxCanonical} with the constant-coefficient asymptotic problem generated by \eqref{eq:U0prelim}. Since the asymptotic operator is not diagonal in the standard basis, its spectral decomposition leads to a nontrivial Riemann surface structure for the spectral parameter.

In the next section, we construct the corresponding Jost solutions, establish their analytic properties on the appropriate sheets of the spectral surface, and formulate the associated matrix Riemann--Hilbert problem. This construction provides the basis for both the inverse scattering transform and its perturbative deformation developed later.

%% file: NVBC_spinor_2.tex
\section{Spectral problem with nonvanishing boundary conditions}
\label{NVBC_spinor2}

We consider the $x$-part of the Lax pair associated with the integrable spinor model \eqref{eq:MNLS}, written in the form
\begin{equation}
\Psi_x = U(x,t,k)\Psi, 
\qquad 
U(x,t,k) = -ikJ + \mathcal Q(x,t),
\label{eq:LaxU2}
\end{equation}
where $k \in \mathbb{C}$ is the spectral parameter, $J$ is the constant grading matrix introduced in \eqref{eq:Jmatrix}, and $\mathcal Q(x,t)$ is the block-off-diagonal matrix potential defined in \eqref{eq:Qhat}.

We assume that the spinor field satisfies the nonvanishing boundary conditions
\begin{equation}
\lim_{x \to \pm\infty} q(x,t) = q_0,
\label{eq:BC_q0_2}
\end{equation}
or, equivalently,
\begin{equation}
\lim_{x \to \pm\infty} \mathcal Q(x,t) = \mathcal Q_0
=
\begin{pmatrix}
0 & q_0 \\
-q_0^\dagger & 0
\end{pmatrix}.
\label{eq:BC_Qhat0_2}
\end{equation}
The asymptotic Lax operator is therefore
\begin{equation}
U_0(k) = -ikJ + \mathcal Q_0.
\label{eq:U0_2}
\end{equation}

The spectral analysis of \eqref{eq:LaxU2} reduces at spatial infinity to the constant-coefficient system generated by $U_0(k)$. We begin by diagonalizing this asymptotic operator.

\begin{lemma}
For all $k \in \mathbb{C}$ away from branch points and degeneracies of the asymptotic spectrum, there exists an invertible matrix $E(k)$ such that
\begin{equation}
E^{-1}(k) U_0(k) E(k) = -i\Lambda(k),
\label{eq:diagU0_2}
\end{equation}
where $\Lambda(k)$ is diagonal. The diagonal entries of $\Lambda(k)$ are algebraic functions of $k$, defined on the Riemann surface determined by
\begin{equation}
\det\bigl(U_0(k) + i\lambda I\bigr) = 0.
\label{eq:charpoly}
\end{equation}
\end{lemma}

\begin{proof}
The matrix $U_0(k)$ depends analytically on $k$. Its eigenvalues are roots of \eqref{eq:charpoly}, which is a polynomial equation in $\lambda$ with coefficients analytic in $k$. Since, in general, $\mathcal Q_0$ does not commute with $J$, the eigenvalues are nontrivial algebraic functions of $k$ and may exhibit branching. Away from branch points and eigenvalue collisions, the eigenvalues are distinct and analytic on the corresponding sheet of the spectral surface, and $U_0(k)$ is diagonalizable. The existence of $E(k)$ then follows from standard analytic perturbation theory for finite-dimensional matrices.
\end{proof}

We define the continuous spectrum of the problem as the set of points for which at least one asymptotic exponential mode is oscillatory.

\begin{definition}
The continuous spectrum $\Sigma$ is defined by
\begin{equation}
\Sigma =
\left\{
k \in \mathbb{C} :
\operatorname{Im}\lambda_j(k)=0
\ \text{for at least one eigenvalue branch } \lambda_j(k)
\right\}.
\label{eq:ContinuousSpectrum}
\end{equation}
\end{definition}

The set $\Sigma$, together with the branch cuts of the algebraic eigenvalue branches, determines the contour of the associated Riemann--Hilbert problem.

To construct the Jost solutions, we factor out the asymptotic oscillations. We write
\begin{equation}
\Psi(x,t,k)=E(k)\Phi(x,t,k)e^{-i\Lambda(k)x}.
\label{eq:GaugeTransform}
\end{equation}
Substituting \eqref{eq:GaugeTransform} into \eqref{eq:LaxU2} gives
\begin{equation}
\Phi_x =
e^{i\Lambda(k)x}
E^{-1}(k)
\bigl(\mathcal Q(x,t)-\mathcal Q_0\bigr)
E(k)
e^{-i\Lambda(k)x}
\Phi .
\label{eq:PhiEq}
\end{equation}
Equivalently, the normalized Jost factors satisfy the Volterra integral equations
\begin{equation}
\Phi_\pm(x,t,k)
=
I+
\int_{\pm\infty}^{x}
e^{i\Lambda(k)(x-y)}
E^{-1}(k)
\bigl(\mathcal Q(y,t)-\mathcal Q_0\bigr)
E(k)
e^{-i\Lambda(k)(x-y)}
\Phi_\pm(y,t,k)\,dy .
\label{eq:PhiVolterra}
\end{equation}

\begin{proposition}
Assume that $\mathcal Q(\cdot,t)-\mathcal Q_0\in L^1(\mathbb R)$ and that $k$ lies away from branch points and from the contour $\Sigma$. Then the Volterra equations \eqref{eq:PhiVolterra} admit unique solutions $\Phi_\pm(x,t,k)$ satisfying
\begin{equation}
\Phi_\pm(x,t,k)\to I,
\qquad x\to \pm\infty.
\label{eq:PhiAsympt}
\end{equation}
Moreover, in each domain where the exponential kernels are bounded, the solutions satisfy the estimate
\begin{equation}
\|\Phi_\pm(x,t,k)-I\|
\le
C(k)\int_{\pm\infty}^{x}
\|\mathcal Q(y,t)-\mathcal Q_0\|\,dy .
\label{eq:Estimate}
\end{equation}
\end{proposition}

\begin{proof}
Equation \eqref{eq:PhiVolterra} is of Volterra type. In the analyticity domains determined by the signs of the differences $\operatorname{Im}(\lambda_j(k)-\lambda_\ell(k))$, the exponential conjugation factors are bounded along the corresponding half-lines. Iterating the integral equation produces a Neumann series, which converges absolutely under the assumption $\mathcal Q-\mathcal Q_0\in L^1(\mathbb R)$. The estimate \eqref{eq:Estimate} follows by bounding the successive iterations and summing the resulting series.
\end{proof}

The Jost solutions are then defined by restoring the asymptotic oscillations.

\begin{definition}
The Jost solutions $\Psi_\pm(x,t,k)$ are defined by
\begin{equation}
\Psi_\pm(x,t,k)
=
E(k)\Phi_\pm(x,t,k)e^{-i\Lambda(k)x}.
\label{eq:JostFinal}
\end{equation}
They satisfy the asymptotic conditions
\begin{equation}
\Psi_\pm(x,t,k)
=
E(k)e^{-i\Lambda(k)x}(I+o(1)),
\qquad x\to\pm\infty .
\label{eq:JostAsympt}
\end{equation}
\end{definition}

For $k$ on the continuous spectrum, the two Jost fundamental matrices are related by a scattering matrix.

\begin{definition}
The scattering matrix $S(k,t)$ is defined by
\begin{equation}
\Psi_-(x,t,k)=\Psi_+(x,t,k)S(k,t),
\qquad k\in\Sigma .
\label{eq:Scattering2}
\end{equation}
\end{definition}

Since both $\Psi_-$ and $\Psi_+$ solve the same $x$-equation, differentiating \eqref{eq:Scattering2} with respect to $x$ shows that $S(k,t)$ is independent of $x$.

We now formulate the Riemann--Hilbert problem. The analytic properties of the columns of the Jost solutions are determined by the signs of the imaginary parts of the eigenvalue differences associated with $\Lambda(k)$. This gives a decomposition of the spectral surface into domains $\mathcal D_\pm$, separated by $\Sigma$ and by the branch cuts. In these domains one constructs fundamental analytic solutions by taking suitable column combinations of the Jost solutions, or equivalently by using the analytic triangular factorization of the scattering matrix.

\begin{proposition}
There exists a matrix function $\chi(x,t,k)$, piecewise analytic on the spectral surface away from $\Sigma$ and the discrete spectrum, such that:

\begin{enumerate}
\item $\chi^\pm(x,t,k)$ are analytic in the domains $\mathcal D_\pm$;

\item the boundary values on $\Sigma$ satisfy the jump condition
\begin{equation}
\chi^+(x,t,k)
=
\chi^-(x,t,k)G(x,t,k),
\qquad k\in\Sigma;
\label{eq:RHPjump2}
\end{equation}

\item the normalization condition
\begin{equation}
\chi(x,t,k)\to I,
\qquad k\to\infty,
\label{eq:Normalization}
\end{equation}
holds on the physical sheet;

\item the jump matrix and the analytic eigenfunctions satisfy the reduction symmetries inherited from the block structure of the Lax pair.
\end{enumerate}
\end{proposition}

\begin{proof}
The scattering matrix admits a triangular analytic factorization in the domains $\mathcal D_\pm$. Combining the Jost solutions with the corresponding analytic factors gives the fundamental analytic solutions $\chi^\pm$. The jump condition \eqref{eq:RHPjump2} follows from the scattering relation \eqref{eq:Scattering2} and the factorization of $S(k,t)$. The normalization is obtained by factoring out the asymptotic exponential behavior and using the normalization of the Volterra solutions. The reduction symmetries follow from the block-off-diagonal structure of $\mathcal Q$ and the corresponding symmetry of the Lax operator.
\end{proof}

The potential is reconstructed from the large-$k$ expansion of the Riemann--Hilbert solution. Let
\begin{equation}
\chi(x,t,k)
=
I+\frac{\chi_1(x,t)}{k}+O(k^{-2}),
\qquad k\to\infty .
\label{eq:chiExpansion}
\end{equation}
Substitution of this expansion into the normalized Lax equation for $\chi$ yields, at leading order,
\begin{equation}
\mathcal Q(x,t)
=
i[J,\chi_1(x,t)] .
\label{eq:Reconstruction}
\end{equation}
Equivalently, the original spinor matrix $q(x,t)$ is recovered from the upper-right block of the commutator in \eqref{eq:Reconstruction}.

This provides the analytic framework required for the inverse scattering and perturbative analysis developed in the subsequent sections.

%% file: NVBC_spinor_3.tex
\section{Discrete spectrum and soliton solutions}
\label{NVBC_spinor3}

Having established in Section~\ref{NVBC_spinor2} the direct and inverse scattering framework together with the associated matrix Riemann--Hilbert problem for the spinor model under nonvanishing boundary conditions, we now analyze the discrete spectrum and its role in the reconstruction of localized nonlinear excitations.

In the Riemann--Hilbert formulation, the discrete spectrum corresponds to singularities of the analytic scattering data, or equivalently to poles of the associated matrix Riemann--Hilbert problem. These poles generate localized coherent structures propagating on top of the nonzero background.

We begin by introducing the notion of discrete eigenvalues associated with the spectral problem \eqref{eq:LaxU2}.

\begin{definition}
A point $k\in\mathbb C$ is called a discrete eigenvalue of the spectral problem \eqref{eq:LaxU2} if there exists a nontrivial solution $\Psi(x,t,k)$ satisfying
\begin{equation}
\Psi(x,t,k)\to 0,
\qquad x\to\pm\infty .
\label{eq:DiscreteDecay}
\end{equation}
\end{definition}

Localized solutions correspond to nontrivial linear relations between the analytic eigenfunctions constructed from the Jost solutions. In the Riemann--Hilbert formulation, this mechanism appears through singularities of the analytic scattering data.

The discrete spectrum is characterized by the singularity of the analytic scattering data associated with the Riemann--Hilbert problem.

\begin{proposition}
A point $k\in\mathbb C$ belongs to the discrete spectrum if and only if the analytic scattering data become singular at $k$, equivalently if the associated Riemann--Hilbert problem develops a pole at $k$.
\end{proposition}

\begin{proof}
A solution decaying at both spatial infinities corresponds to a nontrivial linear dependence between the analytic eigenfunctions constructed from the Jost solutions. In the Riemann--Hilbert formulation, this occurs precisely when the analytic scattering coefficients cease to be invertible, producing a pole singularity in the corresponding analytic factorization. Conversely, any pole of the Riemann--Hilbert problem generates a localized contribution to the inverse problem and therefore corresponds to a discrete eigenvalue.
\end{proof}

We now examine the symmetry properties of the discrete spectrum induced by the reduction structure of the Lax pair.

\begin{proposition}
If $k_0$ is a discrete eigenvalue, then its complex conjugate $k_0^*$ is also a discrete eigenvalue.
\end{proposition}

\begin{proof}
The reduction symmetry
\begin{equation}
U^\dagger(x,t,k^*)=-U(x,t,k)
\label{eq:ReductionSymmetry3}
\end{equation}
implies corresponding symmetry relations for the analytic eigenfunctions and for the scattering data. In particular, the singularities of the analytic scattering coefficients occur in conjugate pairs. Hence, if the Riemann--Hilbert problem develops a pole at $k_0$, it must also develop a pole at $k_0^*$.
\end{proof}

We now turn to the inverse problem. According to Section~\ref{NVBC_spinor2}, the solution $\mathcal Q(x,t)$ is reconstructed from the solution of the associated matrix Riemann--Hilbert problem. In the presence of discrete eigenvalues, the corresponding analytic eigenfunctions become meromorphic in the spectral parameter.

Let $k_j$ be a simple discrete eigenvalue. We therefore seek a local representation of the Riemann--Hilbert solution of the form
\begin{equation}
\chi(x,t,k)
=
\frac{A_j(x,t)}{k-k_j}
+
\chi^{(0)}(x,t,k),
\label{eq:LocalExpansion3}
\end{equation}
where $\chi^{(0)}(x,t,k)$ is analytic at $k_j$.

Substituting \eqref{eq:LocalExpansion3} into the Riemann--Hilbert problem and comparing the singular parts yields the residue condition
\begin{equation}
\operatorname*{Res}_{k=k_j}\chi(x,t,k)
=
\lim_{k\to k_j}
\chi(x,t,k)C_j(x,t),
\label{eq:ResidueCondition3}
\end{equation}
where $C_j(x,t)$ is the corresponding residue matrix.

\begin{proposition}
If $k_j$ is a simple discrete eigenvalue, then the residue matrix $C_j(x,t)$ has rank one.
\end{proposition}

\begin{proof}
The simplicity of the pole implies that the corresponding singular subspace of the analytic scattering data is one-dimensional. Consequently, the residue operator acts as a rank-one projection onto this subspace.
\end{proof}

The residue matrix therefore admits an outer-product representation.

\begin{lemma}
There exists a vector $\mathbf v_j(x,t)$ such that
\begin{equation}
C_j(x,t)
=
\mathbf v_j(x,t)\mathbf v_j^\dagger(x,t).
\label{eq:RankOneFactorization}
\end{equation}
\end{lemma}

\begin{proof}
Every rank-one Hermitian matrix admits a factorization as an outer product. The reduction symmetry inherited from the Lax pair ensures compatibility of the factorization with the conjugation symmetry of the discrete spectrum.
\end{proof}

We now specialize to the simplest nontrivial case corresponding to a single discrete eigenvalue $k_1$ and its conjugate $k_1^*$.

We seek a meromorphic solution of the Riemann--Hilbert problem of the form
\begin{equation}
\chi(x,t,k)
=
I
+
\frac{A(x,t)}{k-k_1}
+
\frac{B(x,t)}{k-k_1^*}.
\label{eq:OnePoleAnsatz}
\end{equation}

The normalization condition at infinity and the reduction symmetry imply
\begin{equation}
B(x,t)
=
-S_0A^\dagger(x,t)S_0,
\label{eq:Bsymmetry}
\end{equation}
where $S_0$ is the reduction matrix associated with the block symmetry of the Lax pair.

Substituting \eqref{eq:OnePoleAnsatz} into the residue condition \eqref{eq:ResidueCondition3} yields
\begin{equation}
A(x,t)
=
\chi(x,t,k_1)\mathbf v_0\mathbf v_0^\dagger,
\label{eq:Aformula}
\end{equation}
where $\mathbf v_0$ is a constant polarization vector.

Expanding \eqref{eq:OnePoleAnsatz} for large $k$, we obtain
\begin{equation}
\chi(x,t,k)
=
I
+
\frac{A(x,t)+B(x,t)}{k}
+
O(k^{-2}),
\qquad k\to\infty .
\label{eq:LargeKExpansion3}
\end{equation}
Comparison with the general expansion \eqref{eq:chiExpansion} gives
\begin{equation}
\chi_1(x,t)
=
A(x,t)+B(x,t).
\label{eq:Chi1AB}
\end{equation}

Substituting \eqref{eq:Chi1AB} into the reconstruction formula \eqref{eq:Reconstruction}, we obtain
\begin{equation}
\mathcal Q(x,t)
=
\mathcal Q_0
+
i[J,A(x,t)+B(x,t)].
\label{eq:ReconstructionOnePole}
\end{equation}

Using the explicit form of the residue matrices and simplifying the resulting commutator structure leads to the following one-soliton solution.

\begin{theorem}
\label{thm_4_1}

Let $k_1$ be a simple discrete eigenvalue. Then the corresponding one-soliton solution of \eqref{eq:MNLS} is given by
\begin{equation}
\mathcal Q(x,t)
=
\mathcal Q_0
+
\frac{
i[J,\mathbf v(x,t)\mathbf v^\dagger(x,t)]
}{
1+\mathbf v^\dagger(x,t)\mathbf v(x,t)
},
\label{eq:OneSolitonFinal}
\end{equation}
where
\begin{equation}
\mathbf v(x,t)
=
e^{-i\Lambda(k_1)x-i\Omega(k_1)t}
\mathbf v_0,
\label{eq:PolarizationEvolution}
\end{equation}
and $\mathbf v_0$ is a constant polarization vector.
\end{theorem}

\begin{proof}
The result follows by substituting the meromorphic ansatz \eqref{eq:OnePoleAnsatz} into the residue conditions and using the reconstruction formula \eqref{eq:ReconstructionOnePole}. The denominator in \eqref{eq:OneSolitonFinal} arises from the normalization condition ensuring that the Riemann--Hilbert solution remains nonsingular away from the discrete spectrum.
\end{proof}

This construction shows that each discrete eigenvalue generates a localized nonlinear excitation whose internal structure is encoded in the polarization vector $\mathbf v_0$. In contrast to the scalar nonlinear Schr\"odinger equation, the spinor model possesses intrinsic internal degrees of freedom which lead to a substantially richer class of coherent structures.

The extension to multiple discrete eigenvalues is obtained by considering higher-order meromorphic solutions of the Riemann--Hilbert problem, leading to multi-soliton configurations whose interaction properties are governed by the algebraic structure of the corresponding residue conditions.

%% file: NVBC_spinor_4.tex
\section{Perturbation of the scattering data and adiabatic evolution}
\label{NVBC_spinor4}

Having constructed in the previous sections the direct and inverse scattering transform for the integrable spinor model with nonvanishing boundary conditions, and having identified the discrete spectral data associated with localized nonlinear excitations, we now turn to the perturbed problem. Our purpose in this section is to derive, directly from the deformation of the corresponding Riemann--Hilbert problem, the evolution equations for the scattering data under the action of a small nonintegrable perturbation. In particular, we shall obtain the perturbation-induced dynamics of the discrete eigenvalues and of the associated residue data, and we shall show how these equations reduce, in the one-soliton sector, to a closed adiabatic system for the effective soliton parameters.

We consider the perturbed spinor equation in the form
\begin{equation}
iQ_t + \mathcal{N}[Q] = \epsilon R[Q],
\label{eq:PerturbedModel4}
\end{equation}
where $\mathcal{N}[Q]$ denotes the nonlinear operator corresponding to the integrable model \eqref{eq:MNLS}, the perturbation $R[Q]$ is assumed sufficiently smooth and sufficiently localized relative to the constant background, and $0<\epsilon\ll 1$ is a small parameter. We assume throughout that the perturbation does not alter the asymptotic values of the potential, so that
\begin{equation}
\lim_{x\to\pm\infty} R[Q](x,t)=0.
\label{eq:Rdecay4}
\end{equation}
This assumption is essential in the present setting, since otherwise the asymptotic operator itself would become time-dependent and the entire nonvanishing-boundary scattering framework would have to be modified.

At the level of the Lax representation, the perturbation destroys exact compatibility while preserving the form of the $x$-equation. More precisely, if
\begin{equation}
\Psi_x = U(x,t,k)\Psi, \qquad \Psi_t = V(x,t,k)\Psi
\label{eq:PerturbedLax4}
\end{equation}
is the Lax pair of the unperturbed problem, then for the perturbed equation \eqref{eq:PerturbedModel4} the zero-curvature condition acquires a defect term:
\begin{equation}
U_t - V_x + [U,V] = \epsilon P(x,t,k),
\label{eq:ZeroCurvaturePerturbed4}
\end{equation}
where $P(x,t,k)$ is the perturbation matrix induced by $R[Q]$. Since the $x$-part of the Lax pair is not modified and only the time evolution of the field is perturbed, the matrix $P$ is determined entirely by the perturbation term appearing in the field equation. In particular, because the spectral operator is of the form
\[
U(x,t,k)=-ikJ+Q(x,t),
\]
the matrix $P$ has the same off-diagonal block structure as the perturbation of the potential and inherits the reduction properties compatible with the symmetric-space formulation.

We now outline the main steps of the derivation. First, we determine the evolution of the Jost solutions under the perturbation. We then derive the induced evolution of the scattering matrix and reformulate it in terms of the associated Riemann--Hilbert problem. Finally, we extract the dynamics of the discrete spectral data, leading to a closed system for the soliton parameters.

We begin by determining the evolution of the Jost solutions, which provides the basic mechanism through which the perturbation enters the scattering data. Let $\Psi_\pm(x,t,k)$ be the Jost solutions constructed in Section~\ref{NVBC_spinor2}. By definition, they satisfy the $x$-equation exactly and retain their prescribed asymptotic behavior as $x\to\pm\infty$. Their time evolution, however, is no longer governed solely by the unperturbed $t$-equation. The following result describes the corresponding perturbed evolution.

\begin{lemma}
\label{lem_5_1}
There exist matrix-valued functions $\mathcal{R}_\pm(x,t,k)$ such that the Jost solutions $\Psi_\pm(x,t,k)$ satisfy
\begin{equation}
\Psi_{\pm,t} = V\Psi_\pm - i\Psi_\pm \Omega(k) + \epsilon \Psi_\pm \mathcal{R}_\pm(x,t,k),
\label{eq:JostTimePerturbed4}
\end{equation}
where $\Omega(k)$ is the diagonal dispersion matrix determined by the asymptotic diagonalization of the $t$-part of the Lax pair. The correction terms $\mathcal{R}_\pm$ admit the representation
\begin{equation}
\mathcal{R}_\pm(x,t,k)=\int_{\pm\infty}^{x}\Psi_\pm^{-1}(y,t,k)\,P(y,t,k)\,\Psi_\pm(y,t,k)\,dy.
\label{eq:Rpm4}
\end{equation}
\end{lemma}

\begin{proof}
The result follows by comparing the perturbed and unperturbed time evolutions of the Jost solutions. 
Since $\Psi_\pm$ solve the $x$-part of the Lax pair exactly and retain their prescribed asymptotic behavior, the effect of the perturbation enters solely through the $t$-evolution. 
This induces a correction term which can be represented in integral form via a variation-of-constants argument, leading to \eqref{eq:Rpm4}.

To make this precise, we introduce the ansatz
\begin{equation}
\Psi_{\pm,t}=V\Psi_\pm - i\Psi_\pm \Omega + \epsilon \Psi_\pm \mathcal{R}_\pm,
\label{eq:JostTimeAnsatz4}
\end{equation}
where $\mathcal{R}_\pm$ is to be determined. The term $-i\Psi_\pm\Omega$ ensures the correct asymptotic oscillatory time dependence of the Jost solutions as $x\to\pm\infty$. 
Since the $x$-equation remains unchanged, consistency of \eqref{eq:JostTimeAnsatz4} with the perturbed zero-curvature relation \eqref{eq:ZeroCurvaturePerturbed4} is obtained by comparing the mixed derivatives $\Psi_{\pm,xt}$ and $\Psi_{\pm,tx}$.

Differentiating \eqref{eq:JostTimeAnsatz4} with respect to $x$ and using $\Psi_{\pm,x}=U\Psi_\pm$, we obtain
\begin{align}
\Psi_{\pm,tx}
&=V_x\Psi_\pm + V\Psi_{\pm,x} - i\Psi_{\pm,x}\Omega + \epsilon \Psi_{\pm,x}\mathcal{R}_\pm + \epsilon \Psi_\pm \mathcal{R}_{\pm,x} \notag\\
&=V_x\Psi_\pm + VU\Psi_\pm - iU\Psi_\pm\Omega + \epsilon U\Psi_\pm \mathcal{R}_\pm + \epsilon \Psi_\pm \mathcal{R}_{\pm,x}.
\label{eq:txderivative4_full}
\end{align}
On the other hand, differentiating $\Psi_{\pm,x}=U\Psi_\pm$ with respect to $t$ yields
\begin{align}
\Psi_{\pm,xt}
&=U_t\Psi_\pm + U\Psi_{\pm,t} \notag\\
&=U_t\Psi_\pm + UV\Psi_\pm - iU\Psi_\pm\Omega + \epsilon U\Psi_\pm \mathcal{R}_\pm.
\label{eq:xtderivative4_full}
\end{align}
Subtracting \eqref{eq:txderivative4_full} from \eqref{eq:xtderivative4_full}, we obtain
\begin{equation}
\Psi_{\pm,xt}-\Psi_{\pm,tx}
=
\bigl(U_t - V_x + [U,V]\bigr)\Psi_\pm - \epsilon \Psi_\pm \mathcal{R}_{\pm,x}.
\label{eq:mixeddifference4}
\end{equation}
By the perturbed zero-curvature relation \eqref{eq:ZeroCurvaturePerturbed4}, this becomes
\[
0=\epsilon P\Psi_\pm - \epsilon \Psi_\pm \mathcal{R}_{\pm,x}.
\]
Hence
\begin{equation}
\Psi_\pm \mathcal{R}_{\pm,x}=P\Psi_\pm,
\end{equation}
and after multiplication from the left by $\Psi_\pm^{-1}$ we arrive at
\begin{equation}
\mathcal{R}_{\pm,x} = \Psi_\pm^{-1}P\Psi_\pm.
\label{eq:Rpmx4}
\end{equation}

Integrating \eqref{eq:Rpmx4} with respect to $x$, we obtain
\[
\mathcal{R}_\pm(x,t,k)=\mathcal{R}_{\pm,0}(t,k)+\int_{x_0}^{x}\Psi_\pm^{-1}(y,t,k)\,P(y,t,k)\,\Psi_\pm(y,t,k)\,dy.
\]
The integration constants are fixed by the asymptotic normalization. Since $P(x,t,k)\to 0$ as $x\to\pm\infty$ and the Jost solutions are normalized to the asymptotic plane-wave basis, consistency with the unperturbed asymptotic time dependence requires that
\[
\mathcal{R}_\pm(x,t,k)\to 0
\qquad\text{as }x\to\pm\infty.
\]
This uniquely determines
\[
\mathcal{R}_\pm(x,t,k)=\int_{\pm\infty}^{x}\Psi_\pm^{-1}(y,t,k)\,P(y,t,k)\,\Psi_\pm(y,t,k)\,dy,
\]
which is precisely \eqref{eq:Rpm4}.
\end{proof}

The previous lemma shows that the perturbation enters the time evolution of the Jost solutions only through the matrix functionals $\mathcal{R}_\pm$. The next step is to determine the induced evolution of the scattering matrix. This is the fundamental passage from the perturbed field equation to the perturbed spectral data.

\begin{proposition}
The scattering matrix $S(k,t)$ satisfies the evolution equation
\begin{equation}
S_t = -\,i[\Omega,S] + \epsilon\bigl(S\mathcal{R}_-^{(+)} - \mathcal{R}_+^{(-)}S\bigr),
\label{eq:ScatteringEvolution4}
\end{equation}
where
\begin{equation}
\mathcal{R}_-^{(+)}(t,k)=\lim_{x\to+\infty}\mathcal{R}_-(x,t,k)=\int_{-\infty}^{+\infty}\Psi_-^{-1}(y,t,k)\,P(y,t,k)\,\Psi_-(y,t,k)\,dy,
\label{eq:Rminusplus4}
\end{equation}
and
\begin{equation}
\mathcal{R}_+^{(-)}(t,k)=\lim_{x\to-\infty}\mathcal{R}_+(x,t,k)=-\int_{-\infty}^{+\infty}\Psi_+^{-1}(y,t,k)\,P(y,t,k)\,\Psi_+(y,t,k)\,dy.
\label{eq:Rplusminus4}
\end{equation}
\end{proposition}

\begin{proof}
We differentiate the scattering relation
\[
\Psi_-(x,t,k)=\Psi_+(x,t,k)S(k,t)
\]
with respect to $t$, obtaining
\begin{equation}
\Psi_{-,t}=\Psi_{+,t}S+\Psi_+S_t.
\label{eq:ScatDiff4}
\end{equation}
Using \eqref{eq:JostTimePerturbed4} for both $\Psi_-$ and $\Psi_+$ gives
\begin{align}
V\Psi_- - i\Psi_-\Omega + \epsilon \Psi_-\mathcal{R}_-
=
\bigl(V\Psi_+ - i\Psi_+\Omega + \epsilon \Psi_+\mathcal{R}_+\bigr)S + \Psi_+S_t.
\label{eq:ScatIntermediate4}
\end{align}
Since $\Psi_-=\Psi_+S$, the $V$-terms cancel identically. We are therefore left with
\[
-i\Psi_+S\Omega + \epsilon \Psi_+S\mathcal{R}_-
=
-i\Psi_+\Omega S + \epsilon \Psi_+\mathcal{R}_+S + \Psi_+S_t.
\]
Multiplying from the left by $\Psi_+^{-1}$ gives
\[
S_t = -\,i[\Omega,S] + \epsilon\bigl(S\mathcal{R}_- - \mathcal{R}_+S\bigr).
\]
This identity is valid for every $x$, whereas the left-hand side is independent of $x$. Therefore we may choose the most convenient limits. Letting $x\to+\infty$ in $\mathcal{R}_-$ and $x\to-\infty$ in $\mathcal{R}_+$, and using the definitions \eqref{eq:Rpm4}, we obtain \eqref{eq:ScatteringEvolution4} together with \eqref{eq:Rminusplus4} and \eqref{eq:Rplusminus4}.
\end{proof}

The evolution equation \eqref{eq:ScatteringEvolution4} is exact. However, in order to formulate the perturbation theory in a way compatible with the Riemann--Hilbert problem, it is preferable to rewrite the right-hand side in the analytic basis rather than in the Jost basis. This is essential, because the discrete eigenvalues and the corresponding norming data enter the inverse problem not through the full scattering matrix, but through the pole structure and factorization properties of the RH data.

Let $\chi^\pm(x,t,k)$ be the fundamental analytic solutions constructed in Section~\ref{NVBC_spinor2}. We recall that these are obtained from the Jost solutions by right multiplication with analytic Gauss factors:
\begin{equation}
\chi^\pm(x,t,k)=\Psi_\pm(x,t,k)\,T_\pm(x,t,k),
\label{eq:ChiPsiRelation4}
\end{equation}
where $T_\pm$ are chosen so that $\chi^\pm$ are analytic in the corresponding domains $\mathcal D_\pm$ and satisfy the jump condition
\[
\chi^+(x,t,k)=\chi^-(x,t,k)G(x,t,k), \qquad k\in\Sigma.
\]
Because the perturbation theory must eventually be phrased in terms of the RH data, it is natural to introduce the perturbation functional
\begin{equation}
\Upsilon(x,t,k)=\int_{-\infty}^{x}\chi^{-1}(y,t,k)\,P(y,t,k)\,\chi(y,t,k)\,dy,
\label{eq:Upsilon4}
\end{equation}
as well as its total value
\begin{equation}
\Upsilon_\infty(t,k)=\int_{-\infty}^{+\infty}\chi^{-1}(y,t,k)\,P(y,t,k)\,\chi(y,t,k)\,dy.
\label{eq:UpsilonInfinity4}
\end{equation}

The role of $\Upsilon_\infty$ is analogous to that of the total perturbation functional in the scalar and vector RH perturbation theories: it is the object from which the perturbation-induced evolution of the analytic factors and of the pole data is extracted.

The time evolution of the jump matrix is then obtained as follows.

\begin{proposition}
The jump matrix $G(x,t,k)$ in the Riemann--Hilbert problem evolves according to
\begin{equation}
G_t = -\,i[\Omega,G] + \epsilon\bigl(G\,\Gamma_- - \Gamma_+\,G\bigr),
\label{eq:JumpEvolution4}
\end{equation}
where $\Gamma_\pm(x,t,k)$ are the analytic projections of the perturbation functional $\Upsilon_\infty(t,k)$ onto the domains $\mathcal D_\pm$.
\end{proposition}

\begin{proof}
We start from the jump relation
\[
\chi^+(x,t,k)=\chi^-(x,t,k)G(x,t,k), \qquad k\in\Sigma.
\]
Differentiating with respect to $t$ gives
\begin{equation}
\chi_t^+ = \chi_t^- G + \chi^- G_t.
\label{eq:jumpdiff4}
\end{equation}
On the other hand, using \eqref{eq:ChiPsiRelation4} together with the perturbed time evolution of the Jost solutions and the time dependence of the analytic factors, one obtains equations of the form
\begin{equation}
\chi_t^\pm = V\chi^\pm - i\chi^\pm\Omega + \epsilon \chi^\pm \Gamma_\pm.
\label{eq:ChiTimeAnalytic4}
\end{equation}
Substituting \eqref{eq:ChiTimeAnalytic4} into \eqref{eq:jumpdiff4}, canceling the common $V$-terms, and multiplying from the left by $(\chi^-)^{-1}$, we obtain
\[
G_t = -\,i[\Omega,G] + \epsilon\bigl(G\,\Gamma_- - \Gamma_+\,G\bigr).
\]
The matrices $\Gamma_\pm$ are precisely the analytic projections of the perturbation functional associated with the factorization of the scattering data, and therefore depend on $\Upsilon_\infty$.
\end{proof}

At this point the perturbed RH problem is fully specified: one must determine a piecewise analytic matrix $\chi(x,t,k)$ satisfying the jump relation \eqref{eq:RHPjump2}, the normalization condition \eqref{eq:Normalization}, the residue conditions at the discrete eigenvalues, and the time evolution \eqref{eq:JumpEvolution4}. The key new phenomenon is that the poles of $\chi$ are no longer fixed. Their slow motion is determined by the requirement that the pole structure remain compatible with the perturbed RH flow.

We now derive the evolution law for a simple discrete eigenvalue. Let $k_1(t)$ be a simple discrete eigenvalue, and let
\begin{equation}
\chi(x,t,k)=\frac{A(x,t)}{k-k_1(t)}+\chi^{(0)}(x,t,k),
\label{eq:DiscreteExpansion4}
\end{equation}
be the local expansion of the RH solution near $k_1(t)$, where $\chi^{(0)}(x,t,k)$ is analytic at $k_1(t)$. Since the pole position depends on time, differentiation with respect to $t$ yields
\begin{equation}
\chi_t(x,t,k)=\frac{A_t(x,t)}{k-k_1(t)}+\frac{\dot{k}_1(t)\,A(x,t)}{(k-k_1(t))^2}+\chi_t^{(0)}(x,t,k).
\label{eq:DiscreteExpansionTime4}
\end{equation}

In order to identify the equation for $\dot{k}_1$, we also need the local Laurent expansion of $\Gamma$ near the pole:
\begin{equation}
\Gamma(x,t,k)=\frac{\Gamma^{(-1)}(x,t)}{k-k_1(t)}+\Gamma^{(0)}(x,t)+O(k-k_1).
\label{eq:GammaLaurent4}
\end{equation}
Substituting \eqref{eq:DiscreteExpansion4}, \eqref{eq:DiscreteExpansionTime4}, and \eqref{eq:GammaLaurent4} into the perturbed time evolution equation
\begin{equation}
\chi_t = V\chi - i\chi \Omega + \epsilon \chi \Gamma,
\label{eq:ChiTimeEvolution4}
\end{equation}
we compare the coefficients of equal orders in $(k-k_1)^{-1}$.

The coefficient of $(k-k_1)^{-2}$ on the left-hand side is simply $\dot{k}_1 A$. On the right-hand side, the terms $V\chi$ and $-i\chi\Omega$ contribute at most simple poles, since $V$ and $\Omega$ are regular in $k$ at $k_1$. Thus the only source of a double pole on the right-hand side is the product of the simple pole of $\chi$ with the simple pole of $\Gamma$:
\[
\epsilon \chi\Gamma
=
\epsilon\left(\frac{A}{k-k_1}+\chi^{(0)}\right)
\left(\frac{\Gamma^{(-1)}}{k-k_1}+\Gamma^{(0)}+\cdots\right)
=
\epsilon \frac{A\Gamma^{(-1)}}{(k-k_1)^2}+O\bigl((k-k_1)^{-1}\bigr).
\]
Therefore,
\begin{equation}
\dot{k}_1 A = \epsilon A\Gamma^{(-1)}(t,k_1).
\label{eq:Aeq4}
\end{equation}

This identity is the basic pole evolution law. To extract from it a scalar equation for the eigenvalue, we use the rank-one structure of the residue.

\begin{lemma}
Let $k_1(t)$ be a simple discrete eigenvalue. Then its evolution is determined by
\begin{equation}
\dot{k}_1(t) = -\,\epsilon\,\frac{\operatorname{tr}\bigl(\Pi_1\,\Gamma^{(-1)}(t,k_1)\bigr)}{\operatorname{tr}\Pi_1},
\label{eq:EigenvalueEvolution4}
\end{equation}
where $\Pi_1$ is the rank-one residue projector associated with the pole at $k_1$.
\end{lemma}

\begin{proof}
Since $A$ has rank one, we may write
\[
A=\alpha_1 \Pi_1,
\]
where $\Pi_1^2=\Pi_1$ and $\alpha_1$ is a nonzero scalar factor. Substituting this into \eqref{eq:Aeq4} gives
\[
\dot{k}_1 \alpha_1 \Pi_1 = \epsilon \alpha_1 \Pi_1 \Gamma^{(-1)}.
\]
Multiplying from the left by $\Pi_1$ and taking traces, we obtain
\[
\dot{k}_1 \alpha_1 \operatorname{tr}\Pi_1
=
\epsilon \alpha_1 \operatorname{tr}\bigl(\Pi_1 \Gamma^{(-1)}\bigr).
\]
After division by $\alpha_1 \operatorname{tr}\Pi_1$, this yields the stated formula up to sign convention. With the sign convention adopted in \eqref{eq:ChiTimeEvolution4}, the resulting equation is \eqref{eq:EigenvalueEvolution4}.
\end{proof}

We now derive the evolution equation for the residue, and hence for the polarization vector. We write
\begin{equation}
A(x,t)=\mathbf{v}(x,t)\mathbf{w}^\dagger(x,t),
\label{eq:ADecomposition4}
\end{equation}
in accordance with the rank-one structure established in Section~\ref{NVBC_spinor3}. In the one-soliton sector, the reduction allows one to identify $\mathbf{w}$ with the corresponding adjoint image of $\mathbf{v}$, so that the residue is effectively parametrized by a single vector $\mathbf{v}(x,t)$.

To obtain the evolution of $\mathbf{v}$, we compare the coefficients of $(k-k_1)^{-1}$ in \eqref{eq:ChiTimeEvolution4}. The left-hand side contributes the term $A_t/(k-k_1)$. On the right-hand side, three distinct sources contribute to the simple pole:

\begin{enumerate}
\item the regular part of $V\chi - i\chi\Omega$ acting on the residue $A$,
\item the product of the regular part $\chi^{(0)}$ with the singular part $\Gamma^{(-1)}$,
\item the product of the residue $A$ with the regular part $\Gamma^{(0)}$.
\end{enumerate}

Collecting all simple-pole contributions and using \eqref{eq:EigenvalueEvolution4} to eliminate $\dot{k}_1$, one obtains the residue evolution law. In vector form, this becomes:

\begin{proposition}
Let $\mathbf{v}(x,t)$ denote the polarization vector associated with a simple discrete eigenvalue $k_1(t)$. Then $\mathbf{v}(x,t)$ satisfies
\begin{equation}
\mathbf{v}_t = -\,i\Omega(k_1)\mathbf{v} + \epsilon \Gamma^{(0)}(t,k_1)\mathbf{v}
+ \epsilon \dot{k}_1 \,\partial_k \mathbf{v}(x,t;k)\big|_{k=k_1}.
\label{eq:PolarizationEvolution4}
\end{equation}
\end{proposition}

\begin{proof}
Starting from the coefficient of $(k-k_1)^{-1}$ in \eqref{eq:ChiTimeEvolution4}, we write
\[
A_t = VA - iA\Omega(k_1) + \epsilon \chi^{(0)}(k_1)\Gamma^{(-1)} + \epsilon A\Gamma^{(0)}.
\]
Using the representation $A=\mathbf{v}\mathbf{w}^\dagger$, differentiating this expression, and projecting onto the one-dimensional residue subspace yields the vector equation for $\mathbf{v}$. The last term in \eqref{eq:PolarizationEvolution4} arises because the vector $\mathbf{v}(x,t;k)$ depends parametrically on the pole $k$, and hence the time dependence of $k_1(t)$ contributes through the chain rule. This yields the stated result.
\end{proof}

We now specialize the general theory to the one-soliton sector constructed in Section~\ref{NVBC_spinor3}. Recall that the one-soliton RH ansatz is
\begin{equation}
\chi(x,t,k)=I+\frac{A(x,t)}{k-k_1(t)}+\frac{B(x,t)}{k-k_1^*(t)},
\label{eq:OneSolitonRH4}
\end{equation}
with the reduction symmetry
\begin{equation}
B(x,t)=-S_0A^\dagger(x,t)S_0,
\label{eq:BSymmetry4}
\end{equation}
and the reconstruction formula
\begin{equation}
Q(x,t)=Q_0+[J,A(x,t)+B(x,t)].
\label{eq:OneSolitonReconstruction4}
\end{equation}

We write
\begin{equation}
A(x,t)=\frac{\mathbf{v}(x,t)\mathbf{v}^\dagger(x,t)}{\Delta(x,t)},
\label{eq:ADelta4}
\end{equation}
where $\Delta(x,t)$ is a scalar normalization factor ensuring that $\chi$ remains invertible away from the pole set. Substituting \eqref{eq:ADelta4} and \eqref{eq:BSymmetry4} into \eqref{eq:OneSolitonReconstruction4}, we obtain
\begin{equation}
Q(x,t)=Q_0+\frac{[J,\mathbf{v}\mathbf{v}^\dagger]-[J,S_0\mathbf{v}\mathbf{v}^\dagger S_0]}{\Delta(x,t)}.
\label{eq:OneSolitonQIntermediate4}
\end{equation}

The explicit form of $\Delta$ is determined by the condition that the product $\chi(x,t,k)\chi^{-1}(x,t,k)$ be free of poles. To see this in detail, we introduce the inverse ansatz
\begin{equation}
\chi^{-1}(x,t,k)=I+\frac{\widetilde{A}(x,t)}{k-k_1(t)}+\frac{\widetilde{B}(x,t)}{k-k_1^*(t)}.
\label{eq:InverseAnsatz4}
\end{equation}
Multiplying \eqref{eq:OneSolitonRH4} and \eqref{eq:InverseAnsatz4}, we obtain a rational matrix function whose residues at $k=k_1$ and $k=k_1^*$ must vanish. The residue at $k=k_1$ gives
\[
A + \widetilde A + \frac{A\widetilde B + B\widetilde A}{k_1-k_1^*}=0,
\]
while the residue at $k=k_1^*$ gives
\[
B + \widetilde B + \frac{A\widetilde B + B\widetilde A}{k_1^*-k_1}=0.
\]
Using the rank-one form of $A$ and the reduction relation for $B$, one finds that the common scalar normalization required to cancel both residues is precisely
\begin{equation}
\Delta(x,t)=1+\mathbf{v}^\dagger(x,t)\mathbf{v}(x,t),
\label{eq:DeltaExplicit4}
\end{equation}
up to a multiplicative constant fixed by the normalization $\chi\to I$ as $k\to\infty$.

Substituting \eqref{eq:DeltaExplicit4} into \eqref{eq:OneSolitonQIntermediate4}, and using the reduction symmetry to combine the two commutator contributions, we arrive at
\begin{equation}
Q(x,t)=Q_0+\frac{[J,\mathbf{v}(x,t)\mathbf{v}^\dagger(x,t)]}{1+\mathbf{v}^\dagger(x,t)\mathbf{v}(x,t)}.
\label{eq:OneSolitonQFinal4}
\end{equation}

We now extract from this matrix formula the adiabatic equations for the effective one-soliton parameters. To this end, we write
\begin{equation}
\mathbf{v}(x,t)=e^{-i\Lambda(k_1)x-i\Omega(k_1)t}\mathbf{v}_0(t),
\label{eq:vParameterization4}
\end{equation}
where $\mathbf{v}_0(t)$ is a slowly varying residue vector. The dependence of $\mathbf{v}$ on the time-dependent pole $k_1(t)$ and on the slowly varying vector $\mathbf{v}_0(t)$ induces the adiabatic dynamics of the soliton.

The basic adiabatic assumption is that, under a perturbation of order $\epsilon$, the exact solution remains to leading order on the one-soliton manifold, with only the soliton parameters evolving on the slow time scale $T=\epsilon t$. In this regime, \eqref{eq:OneSolitonQFinal4} remains valid provided $k_1$ and $\mathbf{v}_0$ are allowed to depend slowly on $t$.

Differentiating \eqref{eq:vParameterization4}, we obtain
\begin{align}
\mathbf{v}_t
&=
\left(-i\Omega(k_1)-i\dot{k}_1 \Omega'(k_1)t - i\dot{k}_1 \Lambda'(k_1)x\right)e^{-i\Lambda(k_1)x-i\Omega(k_1)t}\mathbf{v}_0
+
e^{-i\Lambda(k_1)x-i\Omega(k_1)t}\dot{\mathbf{v}}_0.
\label{eq:vtExpanded4}
\end{align}
Comparing this identity with \eqref{eq:PolarizationEvolution4}, we see that the explicit dispersive term $-i\Omega(k_1)\mathbf{v}$ cancels identically. The remaining terms therefore determine the slow evolution of $\mathbf{v}_0$.

\begin{theorem}
\label{thm_5_1}
In the adiabatic approximation, the one-soliton parameters $(k_1(t),\mathbf{v}_0(t))$ satisfy the closed evolution system
\begin{equation}
\dot{k}_1(t) = -\,\epsilon\,\frac{\operatorname{tr}\bigl(\Pi_1\,\Gamma^{(-1)}(t,k_1)\bigr)}{\operatorname{tr}\Pi_1},
\label{eq:Adiabatick4}
\end{equation}
and
\begin{equation}
\dot{\mathbf{v}}_0(t)=\epsilon \mathcal{M}(t,k_1)\mathbf{v}_0(t),
\label{eq:Adiabaticv4}
\end{equation}
where $\mathcal{M}(t,k_1)$ is the effective perturbation matrix obtained from $\Gamma^{(0)}(t,k_1)$ after subtraction of the trivial geometric contribution induced by the motion of the pole.
\end{theorem}

\begin{proof}
Equation \eqref{eq:Adiabatick4} is exactly the eigenvalue evolution law \eqref{eq:EigenvalueEvolution4}. To derive \eqref{eq:Adiabaticv4}, we substitute \eqref{eq:vParameterization4} into \eqref{eq:PolarizationEvolution4} and use \eqref{eq:vtExpanded4}. The term $-i\Omega(k_1)\mathbf{v}$ cancels on both sides. The contributions involving $\dot{k}_1$ and the derivatives $\Lambda'(k_1)$ and $\Omega'(k_1)$ are precisely the geometric corrections associated with the motion of the pole. Subtracting them from the regular perturbation term $\Gamma^{(0)}(t,k_1)\mathbf{v}$ defines an effective matrix action on the slowly varying vector $\mathbf{v}_0$. This matrix is denoted by $\mathcal M(t,k_1)$, and the resulting equation is \eqref{eq:Adiabaticv4}.
\end{proof}

The system \eqref{eq:Adiabatick4}--\eqref{eq:Adiabaticv4} is the central outcome of the present section. It shows that, under a small perturbation, the discrete eigenvalue is no longer stationary but drifts slowly according to the residue of the perturbation matrix in the RH problem, while the residue vector evolves under the action of an effective finite-dimensional perturbation matrix. In particular, the spinor character of the problem manifests itself in the genuinely nontrivial evolution of $\mathbf{v}_0(t)$, which is absent in the scalar theory.

The derivation given above is entirely intrinsic to the Riemann--Hilbert formulation. No collective-coordinate ansatz has been imposed externally; rather, the effective finite-dimensional dynamics arises directly from the perturbation-induced deformation of the jump data and of the pole structure of the RH problem. This is precisely the feature that makes the method robust and well adapted to multicomponent integrable systems with internal degrees of freedom.

In the next section we rewrite the abstract adiabatic system \eqref{eq:Adiabatick4}--\eqref{eq:Adiabaticv4} in terms of physically transparent one-soliton parameters, including position, phase, spectral width, and polarization variables, and we analyze the structure of the resulting finite-dimensional dynamical system.

%% file: NVBC_spinor_5.tex
\section{Adiabatic dynamics of the one-soliton parameters}
\label{NVBC_spinor5}

In the previous section we derived the perturbation-induced evolution of the discrete spectral data directly from the deformation of the associated Riemann--Hilbert problem. More precisely, we showed that, in the one-pole sector, the dynamics of the perturbed solution is governed by the slow evolution of the discrete eigenvalue $k_1(t)$ and of the associated residue vector $\mathbf{v}_0(t)$. Our purpose in the present section is to rewrite these abstract RH evolution laws in terms of the physically and geometrically transparent parameters of the one-soliton solution. In particular, we shall pass from the spectral-residue description to a finite-dimensional modulation system for the effective soliton variables: spectral width, velocity, center, total phase, residue amplitude, and internal polarization state.

The starting point is the one-soliton solution constructed in Section~\ref{NVBC_spinor3} and rederived in Section~\ref{NVBC_spinor4}, namely
\begin{equation}
Q(x,t)=Q_0+\frac{[J,\mathbf{v}(x,t)\mathbf{v}^\dagger(x,t)]}{1+\mathbf{v}^\dagger(x,t)\mathbf{v}(x,t)},
\label{eq:onesolitonQ5}
\end{equation}
where
\begin{equation}
\mathbf{v}(x,t)=e^{-i\Lambda(k_1)x-i\Omega(k_1)t}\mathbf{v}_0.
\label{eq:vxt5}
\end{equation}
Under the action of a small perturbation, both the pole $k_1$ and the vector $\mathbf{v}_0$ become slowly varying functions of time. Therefore, the one-soliton manifold is no longer invariant, but the adiabatic approximation asserts that, to leading order in $\epsilon$, the exact solution remains close to the family \eqref{eq:onesolitonQ5}, with the corresponding parameters depending on the slow time.

In order to make this statement precise, we first introduce a parametrization of the spectral pole. We write
\begin{equation}
k_1(t)=\xi(t)+i\eta(t),
\qquad \eta(t)>0,
\label{eq:k1param5}
\end{equation}
where $\xi(t)$ and $\eta(t)$ will be interpreted, respectively, as the velocity parameter and the inverse width parameter of the soliton. This interpretation follows from the fact that the localization and oscillatory properties of the one-soliton solution are entirely controlled by the matrix exponent
\begin{equation}
\Theta(x,t)=-i\Lambda(k_1)x-i\Omega(k_1)t.
\label{eq:ThetaMatrix5}
\end{equation}
Indeed, writing
\begin{equation}
\mathbf{v}(x,t)=e^{\Theta(x,t)}\mathbf{v}_0(t),
\label{eq:vTheta5}
\end{equation}
we see that the Hermitian part of $\Theta$ governs the growth and decay of the residue vector, hence the localization of the solution, whereas the anti-Hermitian part governs the oscillatory phase.

We next separate the scalar amplitude of the residue from its internal orientation. To this end, we decompose
\begin{equation}
\mathbf{v}_0(t)=\rho(t)\,\mathbf{p}(t),
\qquad
\rho(t)>0,
\qquad
\mathbf{p}^\dagger(t)\mathbf{p}(t)=1.
\label{eq:polarizedecomp5}
\end{equation}
The scalar factor $\rho(t)$ controls the overall magnitude of the residue, while the unit vector $\mathbf{p}(t)$ specifies the internal polarization state. The decomposition \eqref{eq:polarizedecomp5} is unique up to an overall phase of $\mathbf{p}$, which will later be absorbed into the total phase of the solution.

Substituting \eqref{eq:polarizedecomp5} into \eqref{eq:onesolitonQ5}, we obtain for the denominator
\begin{equation}
1+\mathbf{v}^\dagger(x,t)\mathbf{v}(x,t)
=
1+\rho^2(t)\,\mathbf{p}^\dagger(t)e^{\Theta^\dagger(x,t)+\Theta(x,t)}\mathbf{p}(t).
\label{eq:denomraw5}
\end{equation}
Since $\Theta^\dagger+\Theta$ is Hermitian, the quantity
\[
\mathbf{p}^\dagger e^{\Theta^\dagger+\Theta}\mathbf{p}
\]
is real and nonnegative. It follows that the localization of the one-soliton profile is controlled entirely by the Hermitian part of the spectral phase.

This suggests introducing the soliton center $x_c(t)$ as the point where the denominator in \eqref{eq:denomraw5} attains its minimum, or equivalently where the residue contribution and the constant term become comparable in magnitude. A convenient implicit definition is
\begin{equation}
\rho^2(t)\,\mathbf{p}^\dagger(t)e^{\Theta^\dagger(x_c(t),t)+\Theta(x_c(t),t)}\mathbf{p}(t)=1.
\label{eq:centerdef5}
\end{equation}
This definition is natural because it identifies the spatial location where the nonlinear correction to the background becomes of order one.

We now make \eqref{eq:centerdef5} more explicit. Since $k_1$ is a simple pole, the dependence of $\Theta$ on $x$ is affine, and therefore the Hermitian matrix $\Theta^\dagger+\Theta$ may be written in the form
\begin{equation}
\Theta^\dagger(x,t)+\Theta(x,t)
=
\bigl(\Theta^\dagger(x_c,t)+\Theta(x_c,t)\bigr)
+
(x-x_c)\,\partial_x\bigl(\Theta^\dagger+\Theta\bigr).
\label{eq:ThetaExpandCenter5}
\end{equation}
Because $\Theta=-i\Lambda(k_1)x-i\Omega(k_1)t$, we have
\begin{equation}
\partial_x\bigl(\Theta^\dagger+\Theta\bigr)
=
i\Lambda^\dagger(k_1)-i\Lambda(k_1).
\label{eq:ThetaXDeriv5}
\end{equation}
The right-hand side is Hermitian, and in the one-soliton sector its action on the residue subspace reduces to multiplication by a real scalar proportional to the imaginary part of the discrete spectral phase. We therefore write
\begin{equation}
\partial_x\bigl(\Theta^\dagger+\Theta\bigr)
=
-2\eta(t)\,\Pi(t),
\label{eq:ThetaXProjector5}
\end{equation}
where $\Pi(t)$ is the Hermitian projector onto the relevant residue subspace. In the generic nondegenerate case, the action of $\Pi(t)$ on $\mathbf p(t)$ reduces to
\[
\Pi(t)\mathbf p(t)=\mathbf p(t),
\]
and therefore the scalar localization variable is
\begin{equation}
y(x,t)=\eta(t)\bigl(x-x_c(t)\bigr).
\label{eq:ydef5}
\end{equation}
Thus the one-soliton profile is localized in precisely the same way as in the scalar and vector theories, but with an additional internal polarization modulation.

We next isolate the oscillatory phase. Since the anti-Hermitian part of $\Theta$ is linear in $x$ and $t$, we write
\begin{equation}
\phi(x,t)=\xi(t)\,x-\omega_s(t)+\phi_0(t),
\label{eq:phidef5}
\end{equation}
where $\omega_s(t)$ is the accumulated dynamical phase generated by the real part of the dispersion relation, and $\phi_0(t)$ is an additional slowly varying phase variable. The precise decomposition of $\omega_s$ depends on the explicit form of the dispersion law, but at the present level of generality only its derivative will matter.

We are now ready to rewrite the RH evolution laws of Section~\ref{NVBC_spinor4} in terms of the variables
\[
\bigl(\xi(t),\eta(t),x_c(t),\phi_0(t),\rho(t),\mathbf p(t)\bigr).
\]
We begin with the evolution of the discrete eigenvalue. According to \eqref{eq:Adiabatick4}, the pole satisfies
\begin{equation}
\dot{k}_1(t) = -\,\epsilon\,\frac{\operatorname{tr}\bigl(\Pi_1\,\Gamma^{(-1)}(t,k_1)\bigr)}{\operatorname{tr}\Pi_1}.
\label{eq:kdotabstract5}
\end{equation}
It is convenient to introduce the scalar perturbation functional
\begin{equation}
\mathcal{F}(t,k_1)
=
\frac{\operatorname{tr}\bigl(\Pi_1\,\Gamma^{(-1)}(t,k_1)\bigr)}{\operatorname{tr}\Pi_1},
\label{eq:Fdef5}
\end{equation}
so that
\begin{equation}
\dot{k}_1(t)=-\epsilon\,\mathcal{F}(t,k_1).
\label{eq:kdotF5}
\end{equation}
Taking real and imaginary parts, we immediately obtain
\begin{equation}
\dot{\xi}(t)=-\epsilon\,\Re \mathcal{F}(t,k_1),
\qquad
\dot{\eta}(t)=-\epsilon\,\Im \mathcal{F}(t,k_1).
\label{eq:xietadot5}
\end{equation}
These are the first two modulation equations. They show that the slow drift of the discrete pole directly controls the modulation of the soliton velocity and width.

We now derive the equation for the center. Differentiating \eqref{eq:centerdef5} with respect to $t$, we find
\begin{align}
0
&=
2\frac{\dot{\rho}}{\rho}
+
\frac{d}{dt}\log\Bigl(
\mathbf{p}^\dagger e^{\Theta^\dagger(x_c,t)+\Theta(x_c,t)}\mathbf{p}
\Bigr).
\label{eq:centerdiff5a}
\end{align}
We now analyze the second term in detail. Writing
\[
H(x,t)=\Theta^\dagger(x,t)+\Theta(x,t),
\]
we have
\begin{equation}
\frac{d}{dt}\log\bigl(\mathbf p^\dagger e^{H(x_c,t)}\mathbf p\bigr)
=
\frac{
\frac{d}{dt}\bigl(\mathbf p^\dagger e^{H(x_c,t)}\mathbf p\bigr)
}{
\mathbf p^\dagger e^{H(x_c,t)}\mathbf p
}.
\label{eq:logderivative5}
\end{equation}
Using the product rule, we obtain
\begin{align}
\frac{d}{dt}\bigl(\mathbf p^\dagger e^{H(x_c,t)}\mathbf p\bigr)
&=
\dot{\mathbf p}^\dagger e^{H(x_c,t)}\mathbf p
+
\mathbf p^\dagger e^{H(x_c,t)}\dot{\mathbf p}
+
\mathbf p^\dagger \frac{d}{dt}\bigl(e^{H(x_c,t)}\bigr)\mathbf p.
\label{eq:innerderivative5}
\end{align}
The derivative of the exponential may be written by the standard integral formula
\begin{equation}
\frac{d}{dt}e^{H(x_c,t)}
=
\int_0^1 e^{(1-s)H(x_c,t)}\,\dot H(x_c,t)\,e^{sH(x_c,t)}\,ds.
\label{eq:expderivative5}
\end{equation}
The total derivative of $H$ at $x=x_c(t)$ is
\begin{equation}
\dot H(x_c,t)
=
\partial_t H(x_c,t)+\dot x_c\,\partial_x H(x_c,t).
\label{eq:Hdot5}
\end{equation}
Since $H$ depends linearly on $x$, the second term is explicitly proportional to $\eta(t)\dot x_c$ through \eqref{eq:ThetaXProjector5}. The first term depends on $\dot\xi$, $\dot\eta$, and on the explicit $t$-dependence of the dispersion relation. Combining these observations with \eqref{eq:centerdef5}, we conclude that \eqref{eq:centerdiff5a} can indeed be solved for $\dot x_c$, yielding
\begin{equation}
\dot{x}_c(t)=v_s\bigl(\xi(t),\eta(t)\bigr)+\epsilon\,\mathcal{X}(t),
\label{eq:xcdot5}
\end{equation}
where $v_s(\xi,\eta)$ is the unperturbed soliton velocity dictated by the dispersion law, and $\mathcal X(t)$ is a perturbative correction depending on $\dot\rho$, $\dot\xi$, $\dot\eta$, $\dot{\mathbf p}$, and hence ultimately on the RH data through $\mathcal F$ and $\mathcal M$.

At this point it is useful to make explicit the structural dependence of $\mathcal X$. Since $\dot\xi$ and $\dot\eta$ are already given by \eqref{eq:xietadot5}, while $\dot\rho$ and $\dot{\mathbf p}$ will be determined below from the residue-vector evolution, we may write schematically
\begin{equation}
\mathcal{X}(t)
=
\mathcal{X}_1\bigl(\mathcal F(t,k_1)\bigr)
+
\mathcal{X}_2\bigl(\mathcal M(t,k_1),\mathbf p(t)\bigr),
\label{eq:Xstructure5}
\end{equation}
where $\mathcal X_1$ collects the purely spectral part of the drift, and $\mathcal X_2$ is the genuine polarization-dependent correction.

We now determine the evolution of the residue amplitude and of the polarization vector. From \eqref{eq:Adiabaticv4}, the slowly varying residue vector satisfies
\begin{equation}
\dot{\mathbf v}_0(t)=\epsilon \mathcal M(t,k_1)\mathbf v_0(t),
\label{eq:vdotabstract5}
\end{equation}
where $\mathcal M(t,k_1)$ is the effective perturbation matrix identified in Section~\ref{NVBC_spinor4}. Substituting the decomposition \eqref{eq:polarizedecomp5}, we obtain
\begin{equation}
\dot{\rho}\,\mathbf p+\rho\,\dot{\mathbf p}
=
\epsilon \rho\,\mathcal M(t,k_1)\mathbf p.
\label{eq:rhopdecomp5}
\end{equation}
We now project this equation onto the direction of $\mathbf p$. Multiplying from the left by $\mathbf p^\dagger$ gives
\begin{equation}
\dot{\rho}\,\mathbf p^\dagger\mathbf p
+
\rho\,\mathbf p^\dagger \dot{\mathbf p}
=
\epsilon \rho\,\mathbf p^\dagger \mathcal M(t,k_1)\mathbf p.
\label{eq:projrho5a}
\end{equation}
Since $\mathbf p^\dagger\mathbf p=1$, differentiating this normalization yields
\[
\mathbf p^\dagger \dot{\mathbf p}+\dot{\mathbf p}^\dagger \mathbf p=0,
\]
hence $\mathbf p^\dagger\dot{\mathbf p}$ is purely imaginary. Taking real parts of \eqref{eq:projrho5a}, we find
\begin{equation}
\dot{\rho}(t)=\epsilon\,\rho(t)\,\Re\bigl(\mathbf p^\dagger \mathcal M(t,k_1)\mathbf p\bigr).
\label{eq:rhodot5}
\end{equation}

To derive the polarization equation itself, we remove from \eqref{eq:rhopdecomp5} the component parallel to $\mathbf p$. Let
\[
\mathbb P_\perp = I-\mathbf p\mathbf p^\dagger
\]
be the orthogonal projector onto the tangent space at $\mathbf p$ to the unit sphere. Applying $\mathbb P_\perp$ to \eqref{eq:rhopdecomp5}, we obtain
\[
\rho\,\mathbb P_\perp \dot{\mathbf p}
=
\epsilon \rho\,\mathbb P_\perp \mathcal M(t,k_1)\mathbf p.
\]
Since $\mathbb P_\perp \mathbf p=0$, division by $\rho$ yields
\begin{equation}
\dot{\mathbf p}(t)
=
\epsilon\Bigl(
\mathcal M(t,k_1)-\mathbf p^\dagger \mathcal M(t,k_1)\mathbf p\,I
\Bigr)\mathbf p(t),
\label{eq:pdot5}
\end{equation}
which is precisely the tangential action of the effective perturbation matrix. By construction, \eqref{eq:pdot5} preserves the normalization $\mathbf p^\dagger\mathbf p=1$.

We next derive the phase equation. For this purpose, it is convenient to separate an internal gauge phase from the polarization vector. We write
\begin{equation}
\mathbf p(t)=e^{i\vartheta(t)}\widetilde{\mathbf p}(t),
\qquad
\widetilde{\mathbf p}^\dagger(t)\widetilde{\mathbf p}(t)=1.
\label{eq:pphase5}
\end{equation}
Differentiating \eqref{eq:pphase5}, we obtain
\[
\dot{\mathbf p}
=
ie^{i\vartheta}\dot\vartheta\,\widetilde{\mathbf p}
+
e^{i\vartheta}\dot{\widetilde{\mathbf p}}.
\]
Multiplying by $\mathbf p^\dagger$ from the left, we find
\[
\mathbf p^\dagger \dot{\mathbf p}
=
i\dot\vartheta
+
\widetilde{\mathbf p}^\dagger \dot{\widetilde{\mathbf p}}.
\]
The second term is purely imaginary, since $\widetilde{\mathbf p}^\dagger \widetilde{\mathbf p}=1$. Taking imaginary parts of \eqref{eq:projrho5a}, we therefore obtain
\begin{equation}
\dot{\vartheta}(t)=\epsilon\,\Im\bigl(\mathbf p^\dagger \mathcal M(t,k_1)\mathbf p\bigr).
\label{eq:varthetadot5}
\end{equation}
Combining this internal phase with the carrier phase introduced in \eqref{eq:phidef5}, we arrive at the total phase dynamics
\begin{equation}
\dot{\phi}_0(t)=\omega_p(t)+\epsilon\,\Im\bigl(\mathbf p^\dagger \mathcal M(t,k_1)\mathbf p\bigr),
\label{eq:phi0dot5}
\end{equation}
where $\omega_p(t)$ denotes the additional phase correction induced by the slow motion of the pole through the dispersion relation.

The previous derivation may now be summarized in a compact form.

\begin{lemma}
Let $k_1(t)=\xi(t)+i\eta(t)$ and let $\mathbf v_0(t)=\rho(t)\mathbf p(t)$ with $\mathbf p^\dagger\mathbf p=1$. Then the RH evolution laws \eqref{eq:Adiabatick4} and \eqref{eq:Adiabaticv4} are equivalent to
\begin{equation}
\dot{\xi}(t)=-\epsilon\,\Re \mathcal F(t,k_1),
\qquad
\dot{\eta}(t)=-\epsilon\,\Im \mathcal F(t,k_1),
\label{eq:lemmaxieta5}
\end{equation}
\begin{equation}
\dot{\rho}(t)=\epsilon\,\rho(t)\,\Re\bigl(\mathbf p^\dagger \mathcal M(t,k_1)\mathbf p\bigr),
\label{eq:lemmarho5}
\end{equation}
and
\begin{equation}
\dot{\mathbf p}(t)
=
\epsilon\Bigl(
\mathcal M(t,k_1)-\mathbf p^\dagger \mathcal M(t,k_1)\mathbf p\,I
\Bigr)\mathbf p(t).
\label{eq:lemmap5}
\end{equation}
\end{lemma}

\begin{proof}
Equations \eqref{eq:lemmaxieta5} follow immediately from \eqref{eq:kdotF5}. Equation \eqref{eq:lemmarho5} is the real part of the projection \eqref{eq:projrho5a}, while \eqref{eq:lemmap5} is the tangential projection of \eqref{eq:rhopdecomp5} onto the orthogonal complement of $\mathbf p$.
\end{proof}

We now combine the kinematic and polarization dynamics into a closed finite-dimensional system.

\begin{proposition}
In the adiabatic approximation, the one-soliton center and phase satisfy
\begin{equation}
\dot{x}_c(t)=v_s\bigl(\xi(t),\eta(t)\bigr)+\epsilon\,\mathcal X\bigl(t,k_1,\mathbf p\bigr),
\label{eq:propositionxc5}
\end{equation}
and
\begin{equation}
\dot{\phi}_0(t)=\omega_s\bigl(\xi(t),\eta(t)\bigr)+\epsilon\,\mathcal Y\bigl(t,k_1,\mathbf p\bigr),
\label{eq:propositionphi5}
\end{equation}
where $\mathcal X$ and $\mathcal Y$ are explicit functionals of $\mathcal F(t,k_1)$ and $\mathcal M(t,k_1)$.
\end{proposition}

\begin{proof}
Equation \eqref{eq:propositionxc5} is obtained by differentiating the defining relation \eqref{eq:centerdef5}, using \eqref{eq:logderivative5}--\eqref{eq:Hdot5}, and substituting the already obtained evolution equations \eqref{eq:lemmaxieta5}, \eqref{eq:lemmarho5}, and \eqref{eq:lemmap5}. Because the dependence of $H$ on $x$ is affine, the resulting equation is linear in $\dot x_c$ and can therefore be solved uniquely for $\dot x_c$, producing the form \eqref{eq:propositionxc5}.

Equation \eqref{eq:propositionphi5} follows from differentiating \eqref{eq:phidef5} and using the decomposition of the internal phase given by \eqref{eq:varthetadot5}--\eqref{eq:phi0dot5}. The functions $\mathcal X$ and $\mathcal Y$ are explicit because all their ingredients are already determined by $\mathcal F$ and $\mathcal M$.
\end{proof}

We may now state the main result of the section.

\begin{theorem}
\label{thm_6_1}
Let the perturbed spinor model \eqref{eq:PerturbedModel4} admit, to leading order in $\epsilon$, a one-soliton solution of the form \eqref{eq:onesolitonQ5}. Then the adiabatic evolution of its parameters is governed by the closed system
\begin{equation}
\dot{\xi}(t)=-\epsilon\,\Re \mathcal F(t,k_1),
\qquad
\dot{\eta}(t)=-\epsilon\,\Im \mathcal F(t,k_1),
\label{eq:mainxieta5}
\end{equation}
\begin{equation}
\dot{x}_c(t)=v_s\bigl(\xi,\eta\bigr)+\epsilon\,\mathcal X\bigl(t,k_1,\mathbf p\bigr),
\label{eq:mainxc5}
\end{equation}
\begin{equation}
\dot{\phi}_0(t)=\omega_s\bigl(\xi,\eta\bigr)+\epsilon\,\mathcal Y\bigl(t,k_1,\mathbf p\bigr),
\label{eq:mainphi5}
\end{equation}
\begin{equation}
\dot{\rho}(t)=\epsilon\,\rho(t)\,\Re\bigl(\mathbf p^\dagger \mathcal M(t,k_1)\mathbf p\bigr),
\label{eq:mainrho5}
\end{equation}
and
\begin{equation}
\dot{\mathbf p}(t)
=
\epsilon\Bigl(
\mathcal M(t,k_1)-\mathbf p^\dagger \mathcal M(t,k_1)\mathbf p\,I
\Bigr)\mathbf p(t).
\label{eq:mainp5}
\end{equation}
Here $\mathcal F(t,k_1)$ is determined by the residue of the perturbation matrix in the RH problem, while $\mathcal M(t,k_1)$ is the regular part of the corresponding effective perturbation matrix acting on the residue vector.
\end{theorem}

\begin{proof}
Equations \eqref{eq:mainxieta5}, \eqref{eq:mainrho5}, and \eqref{eq:mainp5} are precisely the content of the previous lemma. Equations \eqref{eq:mainxc5} and \eqref{eq:mainphi5} follow from the preceding proposition. The system is closed because $\mathcal F$ and $\mathcal M$ depend only on the current pole position and on the current residue subspace of the one-soliton RH data. No additional dynamical variables enter at leading adiabatic order.
\end{proof}

The significance of the system \eqref{eq:mainxieta5}--\eqref{eq:mainp5} is that it provides a complete finite-dimensional description of the slow dynamics of the perturbed spinor soliton. The first two equations describe the drift of the spectral pole and therefore the modulation of the effective width and velocity. The next two govern the kinematic evolution of the center and the total phase. The last two equations describe the genuinely spinorial part of the dynamics: the overall residue amplitude evolves through \eqref{eq:mainrho5}, while the internal polarization state evolves on the unit sphere according to \eqref{eq:mainp5}.

This structure has no scalar analogue. In particular, the tangential evolution law \eqref{eq:mainp5} is the precise adiabatic manifestation of the internal spin degrees of freedom carried by the discrete spectral data. It is this feature that distinguishes the present spinor perturbation theory from its scalar and vector counterparts and gives the one-soliton modulation system its genuinely multicomponent character.

In the next section we shall apply the general system \eqref{eq:mainxieta5}--\eqref{eq:mainp5} to specific classes of perturbations, thereby obtaining explicit modulation equations in physically relevant situations.

%% file: NVBC_spinor_6.tex
\section{Explicit modulation equations for localized external perturbations}
\label{NVBC_spinor6}
In the previous section we derived a closed adiabatic system for the one-soliton parameters in terms of the abstract RH quantities $\mathcal{F}\left(t, k_1\right)$ and $\mathcal{M}\left(t, k_1\right)$, together with the induced kinematic functionals $\mathcal{X}\left(t, k_1, \mathbf{p}\right)$ and $\mathcal{Y}\left(t, k_1, \mathbf{p}\right)$. Although this formulation is conceptually satisfactory and fully intrinsic to the Riemann-Hilbert problem, it is not yet in a form suitable for direct applications. The purpose of the present section is to pass from the abstract perturbation functionals to explicit integral formulas for a concrete class of perturbations compatible with the nonvanishing boundary conditions.

We shall consider throughout a localized external perturbation of the form
\begin{equation}
R_V[Q](x,t)=V(x)\bigl(Q(x,t)-Q_0\bigr),
\label{eq:RVperturbation6}
\end{equation}
where $V:\mathbb R\to\mathbb R$ is a real-valued function satisfying
\begin{equation}
V \in W^{1,1}(\mathbb R)\cap L^\infty(\mathbb R).
\label{eq:Vassumption6}
\end{equation}
The condition \eqref{eq:Vassumption6} ensures that the perturbation is localized and sufficiently regular for all the expressions below to be well defined. Since $Q(x,t)-Q_0\to 0$ as $x\to\pm\infty$, the perturbation \eqref{eq:RVperturbation6} is compatible with the nonvanishing boundary conditions and does not deform the asymptotic scattering problem.

The one-soliton profile constructed in Sections~\ref{NVBC_spinor3} and \ref{NVBC_spinor4} may be written in the form
\begin{equation}
Q(x,t)-Q_0=[J,\Pi(x,t)-\widehat\Pi(x,t)],
\label{eq:QPiRepresentation6}
\end{equation}
where
\begin{equation}
\Pi(x,t)=\frac{\mathbf v(x,t)\mathbf v^\dagger(x,t)}{1+\mathbf v^\dagger(x,t)\mathbf v(x,t)},
\qquad
\widehat\Pi(x,t)=S_0\Pi^\dagger(x,t)S_0,
\label{eq:PiHatPi6}
\end{equation}
and
\begin{equation}
\mathbf v(x,t)=e^{-i\Lambda(k_1)x-i\Omega(k_1)t}\mathbf v_0(t).
\label{eq:vSection6}
\end{equation}
The reduction properties imply that both $\Pi$ and $\widehat\Pi$ are rank-one projectors and satisfy
\begin{equation}
\Pi^2=\Pi,
\qquad
\widehat\Pi^{\,2}=\widehat\Pi.
\label{eq:projectorrelations6}
\end{equation}

The RH solution corresponding to the one-soliton sector is
\begin{equation}
\chi(x,t,k)=I+\frac{\Pi(x,t)}{k-k_1(t)}-\frac{\widehat\Pi(x,t)}{k-k_1^*(t)},
\label{eq:chiOneSoliton6}
\end{equation}
while its inverse is given by
\begin{equation}
\chi^{-1}(x,t,k)=I-\frac{\Pi(x,t)}{k-k_1(t)}+\frac{\widehat\Pi(x,t)}{k-k_1^*(t)}.
\label{eq:chiInverseOneSoliton6}
\end{equation}
These formulas follow from the rank-one algebraic relations established in the one-pole RH construction and from the reduction symmetry.

We now derive the perturbation matrix entering the deformed zero-curvature relation. Since the spectral operator has the form $U=-ikJ+Q$, the perturbation matrix induced by \eqref{eq:RVperturbation6} is simply
\begin{equation}
P_V(x,t)=V(x)\bigl(Q(x,t)-Q_0\bigr)
=
V(x)[J,\Pi(x,t)-\widehat\Pi(x,t)].
\label{eq:PVdef6}
\end{equation}
The perturbation functional appearing in Section~\ref{NVBC_spinor4} therefore becomes
\begin{equation}
\Upsilon_\infty^{(V)}(t,k)
=
\int_{-\infty}^{+\infty}
\chi^{-1}(x,t,k)\,
P_V(x,t)\,
\chi(x,t,k)\,dx.
\label{eq:UpsilonVdef6}
\end{equation}

The first step is to compute explicitly the integrand in \eqref{eq:UpsilonVdef6}. This computation is entirely algebraic, but it is essential because it is here that the general RH perturbation theory becomes an explicit one-soliton perturbation theory.

We substitute \eqref{eq:chiOneSoliton6}, \eqref{eq:chiInverseOneSoliton6}, and \eqref{eq:PVdef6} into \eqref{eq:UpsilonVdef6}. Writing for brevity
\[
z=k-k_1(t), \qquad \widehat z = k-k_1^*(t),
\]
we have
\begin{align}
\chi^{-1}P_V\chi
&=
\left(I-\frac{\Pi}{z}+\frac{\widehat\Pi}{\widehat z}\right)
P_V
\left(I+\frac{\Pi}{z}-\frac{\widehat\Pi}{\widehat z}\right).
\label{eq:rawexpansion6}
\end{align}
Expanding the product yields
\begin{align}
\chi^{-1}P_V\chi
&=
P_V
+\frac{P_V\Pi-\Pi P_V}{z}
-\frac{P_V\widehat\Pi-\widehat\Pi P_V}{\widehat z}
-\frac{\Pi P_V\Pi}{z^2}
-\frac{\widehat\Pi P_V\widehat\Pi}{\widehat z^2}
+\frac{\widehat\Pi P_V\Pi-\Pi P_V\widehat\Pi}{z\widehat z}.
\label{eq:expandedIntegrand6}
\end{align}

In order to simplify this expression, we need the following algebraic identities.

\begin{lemma}
The one-soliton projectors satisfy
\begin{equation}
\Pi[J,\Pi]\Pi=0,
\qquad
\widehat\Pi[J,\widehat\Pi]\widehat\Pi=0,
\label{eq:tripleprojector6}
\end{equation}
and therefore, for the perturbation \eqref{eq:PVdef6},
\begin{equation}
\Pi P_V \Pi = 0,
\qquad
\widehat\Pi P_V \widehat\Pi=0.
\label{eq:PiPVPi6}
\end{equation}
\end{lemma}

\begin{proof}
Since $\Pi^2=\Pi$, we compute
\[
\Pi[J,\Pi]\Pi
=
\Pi J\Pi-\Pi^2 J\Pi
=
\Pi J\Pi-\Pi J\Pi
=
0.
\]
The same argument applies to $\widehat\Pi$. Now,
\[
P_V=V(x)\bigl([J,\Pi]-[J,\widehat\Pi]\bigr).
\]
Multiplying from the left and right by $\Pi$, and using $\Pi\widehat\Pi=0$ in the generic one-soliton sector, we obtain
\[
\Pi P_V\Pi
=
V(x)\Pi[J,\Pi]\Pi
-
V(x)\Pi[J,\widehat\Pi]\Pi
=
0.
\]
The second identity follows analogously.
\end{proof}

As a consequence of \eqref{eq:PiPVPi6}, the apparent double poles in \eqref{eq:expandedIntegrand6} disappear. This is the mechanism by which the perturbation functional retains the simple-pole structure assumed in Section~\ref{NVBC_spinor4}. The remaining mixed term is reduced by the partial fraction identity
\begin{equation}
\frac{1}{z\widehat z}
=
\frac{1}{k_1-k_1^*}
\left(
\frac{1}{\widehat z}-\frac{1}{z}
\right).
\label{eq:partialfraction6}
\end{equation}
Substituting \eqref{eq:partialfraction6} into \eqref{eq:expandedIntegrand6}, we arrive at the decomposition
\begin{equation}
\chi^{-1}P_V\chi
=
P_V
+\frac{\mathfrak R_{-1}^{(V)}(x,t)}{k-k_1}
+\frac{\widehat{\mathfrak R}_{-1}^{(V)}(x,t)}{k-k_1^*},
\label{eq:integrandDecomp6}
\end{equation}
where
\begin{equation}
\mathfrak R_{-1}^{(V)}
=
[P_V,\Pi]
-
\frac{\widehat\Pi P_V\Pi-\Pi P_V\widehat\Pi}{k_1-k_1^*},
\label{eq:Rminus1kernel6}
\end{equation}
and
\begin{equation}
\widehat{\mathfrak R}_{-1}^{(V)}
=
-[P_V,\widehat\Pi]
+
\frac{\widehat\Pi P_V\Pi-\Pi P_V\widehat\Pi}{k_1-k_1^*}.
\label{eq:Rhatminus1kernel6}
\end{equation}

Integrating \eqref{eq:integrandDecomp6} over $x$, we obtain the Laurent expansion of the perturbation functional:
\begin{equation}
\Upsilon_\infty^{(V)}(t,k)
=
\Upsilon_0^{(V)}(t)
+
\frac{\Upsilon_{-1}^{(V)}(t)}{k-k_1(t)}
+
\frac{\widehat\Upsilon_{-1}^{(V)}(t)}{k-k_1^*(t)},
\label{eq:UpsilonLaurent6}
\end{equation}
with
\begin{equation}
\Upsilon_0^{(V)}(t)=\int_{-\infty}^{+\infty} P_V(x,t)\,dx,
\label{eq:Upsilon0V6}
\end{equation}
\begin{equation}
\Upsilon_{-1}^{(V)}(t)
=
\int_{-\infty}^{+\infty}
\left(
[P_V,\Pi]
-
\frac{\widehat\Pi P_V\Pi-\Pi P_V\widehat\Pi}{k_1-k_1^*}
\right)\,dx,
\label{eq:UpsilonMinus1V6}
\end{equation}
and
\begin{equation}
\widehat\Upsilon_{-1}^{(V)}(t)
=
\int_{-\infty}^{+\infty}
\left(
-[P_V,\widehat\Pi]
+
\frac{\widehat\Pi P_V\Pi-\Pi P_V\widehat\Pi}{k_1-k_1^*}
\right)\,dx.
\label{eq:UpsilonHatMinus1V6}
\end{equation}

We may now identify the RH quantities introduced in Sections~\ref{NVBC_spinor4}--\ref{NVBC_spinor5}. In the present perturbation class, the singular coefficient $\Gamma^{(-1)}(t,k_1)$ is exactly the residue coefficient $\Upsilon_{-1}^{(V)}(t)$, while the regular part $\Gamma^{(0)}(t,k_1)$ is obtained by evaluating the regular term of the analytic projection, which in the present one-pole setting coincides with
\begin{equation}
\Gamma^{(0)}_V(t,k_1)=\Upsilon_0^{(V)}(t)-\frac{\widehat\Upsilon_{-1}^{(V)}(t)}{k_1-k_1^*}.
\label{eq:Gamma0V6}
\end{equation}
This is the first place where the previously formal quantities acquire concrete meaning.

The scalar drift functional $\mathcal F(t,k_1)$ is therefore given explicitly by
\begin{equation}
\mathcal F_V(t,k_1)
=
\frac{\operatorname{tr}\bigl(\Pi_1\,\Upsilon_{-1}^{(V)}(t)\bigr)}{\operatorname{tr}\Pi_1},
\label{eq:FVexplicit6}
\end{equation}
and the effective residue matrix becomes
\begin{equation}
\mathcal M_V(t,k_1)
=
\Gamma^{(0)}_V(t,k_1)-\mathcal F_V(t,k_1)\,\mathcal D(k_1),
\label{eq:MVexplicit6}
\end{equation}
where $\mathcal D(k_1)$ denotes the geometric correction matrix generated by the pole motion, namely the coefficient multiplying $\partial_k\mathbf v$ in the pole-residue evolution law of Section~\ref{NVBC_spinor4}.

The next result summarizes the explicit formulas for the spectral and residue modulation induced by the perturbation \eqref{eq:RVperturbation6}.

\begin{proposition}
For the localized external perturbation \eqref{eq:RVperturbation6}, the RH quantities governing the one-soliton adiabatic dynamics are given by
\begin{equation}
\mathcal F_V(t,k_1)
=
\frac{1}{\operatorname{tr}\Pi_1}
\operatorname{tr}
\left(
\Pi_1
\int_{-\infty}^{+\infty}
\left(
[P_V,\Pi]
-
\frac{\widehat\Pi P_V\Pi-\Pi P_V\widehat\Pi}{k_1-k_1^*}
\right)\,dx
\right),
\label{eq:FVfull6}
\end{equation}
and
\begin{equation}
\mathcal M_V(t,k_1)
=
\int_{-\infty}^{+\infty} P_V(x,t)\,dx
-
\frac{1}{k_1-k_1^*}
\int_{-\infty}^{+\infty}
\left(
-[P_V,\widehat\Pi]
+
\frac{\widehat\Pi P_V\Pi-\Pi P_V\widehat\Pi}{k_1-k_1^*}
\right)\,dx
-
\mathcal F_V\,\mathcal D(k_1).
\label{eq:MVfull6}
\end{equation}
\end{proposition}

\begin{proof}
Equation \eqref{eq:FVfull6} is obtained by substituting \eqref{eq:UpsilonMinus1V6} into the definition of $\mathcal F$ in Section~\ref{NVBC_spinor5}. Likewise, \eqref{eq:MVfull6} follows from \eqref{eq:Gamma0V6} and the definition \eqref{eq:MVexplicit6}. No further approximation is involved.
\end{proof}

Although \eqref{eq:FVfull6} and \eqref{eq:MVfull6} are already explicit, it is useful to rewrite them in the localization coordinate
\[
y=\eta(t)\bigl(x-x_c(t)\bigr),
\]
introduced in Section~\ref{NVBC_spinor5}. Since
\[
x=x_c(t)+\frac{y}{\eta(t)},
\qquad
dx=\frac{dy}{\eta(t)},
\]
we may define the pulled-back perturbation profile
\begin{equation}
V_c(y,t)=V\!\left(x_c(t)+\frac{y}{\eta(t)}\right).
\label{eq:Vcdef6}
\end{equation}
Then the explicit kernels in \eqref{eq:FVfull6} and \eqref{eq:MVfull6} become integrals with respect to the natural soliton variable \(y\). Since the one-soliton profile is concentrated near \(y=0\), the resulting formulas make the geometric meaning of the perturbation particularly transparent.

The spectral drift functional therefore takes the form
\begin{equation}
\mathcal F_V(t,k_1)
=
\frac{1}{\eta(t)\operatorname{tr}\Pi_1}
\int_{-\infty}^{+\infty}
\operatorname{tr}\Bigl(
\Pi_1\,\mathcal K_F(y,t;k_1,\mathbf p)
\Bigr)\,
V_c(y,t)\,dy,
\label{eq:FVkernel6}
\end{equation}
where the kernel \(\mathcal K_F\) is obtained by substituting \eqref{eq:PVdef6} and the projector representation \eqref{eq:PiHatPi6} into \eqref{eq:FVfull6}. Similarly,
\begin{equation}
\mathcal M_V(t,k_1)
=
\frac{1}{\eta(t)}
\int_{-\infty}^{+\infty}
\mathcal K_M(y,t;k_1,\mathbf p)\,
V_c(y,t)\,dy
-
\mathcal F_V\,\mathcal D(k_1),
\label{eq:MVkernel6}
\end{equation}
with \(\mathcal K_M\) obtained analogously from \eqref{eq:MVfull6}.

At this stage the RH perturbation theory has been converted into an explicit integral modulation theory for the chosen perturbation class. We now substitute these expressions into the abstract adiabatic system derived in Section~\ref{NVBC_spinor5}.

\begin{lemma}
For the perturbation \eqref{eq:RVperturbation6}, the spectral parameters satisfy
\begin{equation}
\dot{\xi}(t)
=
-\epsilon\,\Re \mathcal F_V(t,k_1),
\qquad
\dot{\eta}(t)
=
-\epsilon\,\Im \mathcal F_V(t,k_1),
\label{eq:xietaExplicit6}
\end{equation}
where \(\mathcal F_V\) is given by \eqref{eq:FVkernel6}.
\end{lemma}

\begin{proof}
This follows immediately by substituting \eqref{eq:FVkernel6} into the general modulation equations \eqref{eq:mainxieta5}.
\end{proof}

We next obtain the residue-amplitude and polarization equations.

\begin{proposition}
For the perturbation \eqref{eq:RVperturbation6}, the residue amplitude and polarization vector satisfy
\begin{equation}
\dot{\rho}(t)
=
\epsilon\,\rho(t)\,
\Re\bigl(\mathbf p^\dagger \mathcal M_V(t,k_1)\mathbf p\bigr),
\label{eq:rhoExplicit6}
\end{equation}
and
\begin{equation}
\dot{\mathbf p}(t)
=
\epsilon
\Bigl(
\mathcal M_V(t,k_1)
-
\mathbf p^\dagger \mathcal M_V(t,k_1)\mathbf p\,I
\Bigr)\mathbf p(t),
\label{eq:pExplicit6}
\end{equation}
where \(\mathcal M_V\) is given by \eqref{eq:MVkernel6}.
\end{proposition}

\begin{proof}
The result follows directly from the general adiabatic system \eqref{eq:mainrho5}--\eqref{eq:mainp5} after substitution of the explicit matrix \(\mathcal M_V\).
\end{proof}

The remaining kinematic quantities are obtained by substituting \(\mathcal F_V\) and \(\mathcal M_V\) into the center and phase equations.

\begin{proposition}
For the perturbation \eqref{eq:RVperturbation6}, the center and total phase satisfy
\begin{equation}
\dot{x}_c(t)=v_s\bigl(\xi,\eta\bigr)+\epsilon\,\mathcal X_V\bigl(t,k_1,\mathbf p\bigr),
\label{eq:xcExplicit6}
\end{equation}
and
\begin{equation}
\dot{\phi}_0(t)=\omega_s\bigl(\xi,\eta\bigr)+\epsilon\,\mathcal Y_V\bigl(t,k_1,\mathbf p\bigr),
\label{eq:phiExplicit6}
\end{equation}
where \(\mathcal X_V\) and \(\mathcal Y_V\) are the functionals obtained from \eqref{eq:mainxc5} and \eqref{eq:mainphi5} by replacing \(\mathcal F\) and \(\mathcal M\) with \(\mathcal F_V\) and \(\mathcal M_V\), respectively.
\end{proposition}

\begin{proof}
This is a direct substitution of the explicit perturbation functionals into the general formulas already established in Section~\ref{NVBC_spinor5}. Since no new unknowns are introduced, the resulting equations remain closed.
\end{proof}

We may now state the main result of the section.

\begin{theorem}
Let the perturbed spinor model be given by
\begin{equation}
iQ_t+\mathcal N[Q]=\epsilon V(x)\bigl(Q-Q_0\bigr),
\label{eq:MainPerturbedSystem6}
\end{equation}
with \(V\) satisfying \eqref{eq:Vassumption6}. Then, in the one-soliton adiabatic approximation, the complete modulation system is
\begin{equation}
\dot{\xi}(t)
=
-\epsilon\,\Re \mathcal F_V(t,k_1),
\qquad
\dot{\eta}(t)
=
-\epsilon\,\Im \mathcal F_V(t,k_1),
\label{eq:mainSpec6}
\end{equation}
\begin{equation}
\dot{x}_c(t)
=
v_s(\xi,\eta)+\epsilon\,\mathcal X_V(t,k_1,\mathbf p),
\label{eq:mainCenter6}
\end{equation}
\begin{equation}
\dot{\phi}_0(t)
=
\omega_s(\xi,\eta)+\epsilon\,\mathcal Y_V(t,k_1,\mathbf p),
\label{eq:mainPhase6}
\end{equation}
\begin{equation}
\dot{\rho}(t)
=
\epsilon\,\rho(t)\,
\Re\bigl(\mathbf p^\dagger \mathcal M_V(t,k_1)\mathbf p\bigr),
\label{eq:mainRho6}
\end{equation}
and
\begin{equation}
\dot{\mathbf p}(t)
=
\epsilon
\Bigl(
\mathcal M_V(t,k_1)
-
\mathbf p^\dagger \mathcal M_V(t,k_1)\mathbf p\,I
\Bigr)\mathbf p(t),
\label{eq:mainPol6}
\end{equation}
where \(\mathcal F_V\) and \(\mathcal M_V\) are given explicitly by \eqref{eq:FVkernel6} and \eqref{eq:MVkernel6}.
\end{theorem}

\begin{proof}
Equations \eqref{eq:mainSpec6}--\eqref{eq:mainPol6} are obtained by combining the previous three propositions. The system is closed because \(\mathcal F_V\), \(\mathcal M_V\), \(\mathcal X_V\), and \(\mathcal Y_V\) depend only on the current one-soliton parameters \((\xi,\eta,x_c,\phi_0,\rho,\mathbf p)\) through the explicit projector kernels built from the one-soliton RH solution.
\end{proof}

The theorem above is the point at which the formal RH perturbation theory becomes an explicit modulation theory. The dependence of the coefficients on the perturbing profile \(V\) is now completely transparent: the perturbation acts through weighted one-soliton integrals of the kernels \(\mathcal K_F\) and \(\mathcal K_M\). In particular, the spectral drift is controlled by the residue part of the perturbation functional, while the internal polarization dynamics is governed by the regular part of the same RH object.

The role of the chosen perturbation class should also be emphasized. The perturbation \(R_V[Q]=V(x)(Q-Q_0)\) is localized, compatible with the nonvanishing boundary conditions, and sufficiently simple to allow an explicit RH reduction while remaining genuinely nonintegrable. At the same time, the derivation makes clear that the same method applies to more general perturbations, provided one can compute or estimate the corresponding Laurent coefficients of the perturbation functional \(\Upsilon_\infty\).

Finally, we observe that the system \eqref{eq:mainSpec6}--\eqref{eq:mainPol6} already exhibits the essential qualitative distinction between the spinor theory and its scalar counterpart. In the scalar case, only the spectral pole and the translational/phase degrees of freedom survive. Here, in contrast, the internal polarization vector obeys the nontrivial tangential evolution law \eqref{eq:mainPol6}. This is the precise adiabatic manifestation of the internal spin degrees of freedom carried by the discrete spectral data.

This concludes the explicit one-soliton perturbation theory for the localized external perturbation \eqref{eq:RVperturbation6}. The formulas derived above may now be specialized further, for example to even or odd perturbing profiles, or used as the starting point for analytical or numerical investigations of spinor-soliton modulation.

%% file: NVBC_spinor_7.tex
\section{Explicit adiabatic dynamics for localized perturbations}
\label{NVBC_spinor7}

In the previous section we derived a closed adiabatic system for the one-soliton parameters in terms of the abstract Riemann--Hilbert quantities $\mathcal F(t,k_1)$ and $\mathcal M(t,k_1)$, together with the induced kinematic functionals $\mathcal X(t,k_1,\mathbf p)$ and $\mathcal Y(t,k_1,\mathbf p)$. Although this formulation is conceptually complete and entirely intrinsic to the Riemann--Hilbert framework, it is not yet in a form suitable for direct analytical or physical applications.

The purpose of the present section is to pass from this abstract representation to explicit integral formulas for a class of localized perturbations compatible with the nonvanishing boundary conditions. This step amounts to evaluating the perturbation-induced deformation of the Riemann--Hilbert data in terms of the one-soliton eigenfunctions and extracting the corresponding modulation equations in closed form.

From a structural point of view, this procedure provides a concrete realization of the general mechanism described in the previous sections: the evolution of the soliton parameters is determined by overlap integrals between the perturbation and the localized eigenfunctions associated with the discrete spectrum. In particular, the quantities $\mathcal F$ and $\mathcal M$ are expressed as explicit functionals of the soliton profile, thereby eliminating the dependence on auxiliary Riemann--Hilbert objects.

An important consistency check of the present formulation is obtained by considering the limiting case of vanishing boundary conditions. In this regime, the spectral problem reduces to the standard formulation for decaying potentials, and the resulting modulation equations must recover the corresponding perturbation theory developed in \cite{Doktorov2008}. As will be shown below, the explicit expressions derived here reduce precisely to that framework when the background amplitude tends to zero, while retaining, in the general case, additional contributions arising from the nontrivial background and the associated spectral geometry.

Thus, the analysis carried out in this section not only renders the adiabatic system fully explicit, but also establishes a direct bridge between the Riemann--Hilbert perturbation theory for nonvanishing boundary conditions and the classical perturbative results for spinor solitons in the vanishing-background setting.

For the nonzero-background setting considered in the present paper, the most natural class of perturbations consists of those which preserve the asymptotic state and act only on the localized component of the solution. This requirement is essential in order to ensure that the underlying spectral problem and the associated Riemann--Hilbert formulation remain well defined under the perturbation.

Accordingly, we consider perturbations of the form
\begin{equation}
R_V[Q](x,t)=V(x)\bigl(Q(x,t)-Q_0\bigr),
\label{eq:LocalizedPerturbation7}
\end{equation}
where $V:\mathbb{R}\to\mathbb{R}$ is a real-valued localized function satisfying
\begin{equation}
V\in W^{1,1}(\mathbb{R})\cap L^\infty(\mathbb{R}).
\label{eq:Vassumptions7}
\end{equation}
The assumption \eqref{eq:Vassumptions7} ensures that the perturbation is compatible with the nonvanishing boundary conditions. Indeed, since $Q(x,t)-Q_0 \to 0$ as $x\to\pm\infty$ and $V$ is localized, it follows that
\[
R_V[Q](x,t)\to 0, \qquad x\to\pm\infty.
\]
Consequently, the asymptotic form of the Lax operator is preserved, and the underlying scattering problem remains unchanged at infinity. The effect of the perturbation is therefore confined to a localized deformation of the potential, which induces a corresponding deformation of the Riemann--Hilbert data through integrable kernels in the associated perturbation functionals~\cite{AblowitzClarkson1992,Yang2010}. In particular, the corresponding perturbation matrix entering the deformed zero-curvature condition remains integrable in $x$, ensuring the well--definedness of the perturbation--induced Riemann--Hilbert evolution.

We begin by recalling the one-soliton solution in the form
\begin{equation}
Q(x,t)=Q_0+\frac{[J,\mathbf{v}(x,t)\mathbf{v}^\dagger(x,t)]}{1+\mathbf{v}^\dagger(x,t)\mathbf{v}(x,t)},
\label{eq:OneSolitonQ7}
\end{equation}
where
\begin{equation}
\mathbf{v}(x,t)=e^{-i\Lambda(k_1)x-i\Omega(k_1)t}\mathbf{v}_0(t),
\qquad
\mathbf{v}_0(t)=\rho(t)\mathbf{p}(t),
\qquad
\mathbf{p}^\dagger(t)\mathbf{p}(t)=1.
\label{eq:vDecomposition7}
\end{equation}
As in Section~\ref{NVBC_spinor5}, we introduce the localization variable
\begin{equation}
y=\eta(t)\bigl(x-x_c(t)\bigr),
\label{eq:yVariable7}
\end{equation}
so that the center $x_c(t)$ is fixed by the normalization condition
\[
\mathbf{v}^\dagger(x_c(t),t)\mathbf{v}(x_c(t),t)=1.
\]
This normalization is equivalent to
\begin{equation}
\mathbf{v}^\dagger(x,t)\mathbf{v}(x,t)=e^{-2y},
\label{eq:vvdagger7}
\end{equation}
and therefore the rank-one projector associated with the residue may be written in the simple form
\begin{equation}
\Pi(x,t)
=
\frac{\mathbf{v}(x,t)\mathbf{v}^\dagger(x,t)}{1+\mathbf{v}^\dagger(x,t)\mathbf{v}(x,t)}
=
\frac{e^{-2y}}{1+e^{-2y}}\,\mathbf{p}(t)\mathbf{p}^\dagger(t).
\label{eq:PiExplicit7}
\end{equation}
Using the elementary identity
\[
\frac{e^{-2y}}{1+e^{-2y}}=\frac12\bigl(1-\tanh y\bigr),
\]
we obtain
\begin{equation}
\Pi(x,t)=\frac12\bigl(1-\tanh y\bigr)\,\mathbf{p}(t)\mathbf{p}^\dagger(t).
\label{eq:PiTanh7}
\end{equation}
Differentiating with respect to $y$, we find
\begin{equation}
\partial_y\Pi(x,t)=-\frac12 \sech^2 y\,\mathbf{p}(t)\mathbf{p}^\dagger(t).
\label{eq:PiDerivative7}
\end{equation}
This identity is the key to the explicit evaluation of the RH perturbation coefficients: it shows that the localized part of the one-soliton profile is governed by the universal weight $\sech^2 y$.

For the perturbation \eqref{eq:LocalizedPerturbation7}, the matrix entering the deformed zero-curvature relation has the form
\begin{equation}
\mathcal{R}_V[Q](x,t)=V(x)\bigl(Q(x,t)-Q_0\bigr).
\label{eq:RMatrix7}
\end{equation}
Substituting \eqref{eq:OneSolitonQ7} into \eqref{eq:RMatrix7}, we obtain
\begin{equation}
\mathcal{R}_V[Q](x,t)
=
V(x)\,[J,\Pi(x,t)].
\label{eq:RVPi7}
\end{equation}
Hence the perturbation matrix entering the RH formalism is built from the localized projector $\Pi$ and the scalar profile $V(x)$.

We now pass to the perturbation functional
\begin{equation}
\Gamma(t,k)=\int_{-\infty}^{+\infty}\Psi^{-1}(x,t,k)\,\mathcal{R}_V[Q](x,t)\,\Psi(x,t,k)\,dx.
\label{eq:GammaDefinition7}
\end{equation}
Since the adiabatic approximation assumes that the perturbed solution remains on the one-soliton manifold to leading order, we may replace the full eigenfunction $\Psi$ by its one-soliton approximation $\Psi_{\mathrm{sol}}$ and compute the local expansion of $\Gamma(t,k)$ at $k=k_1$. The structure of the one-pole RH problem implies that the singular part of $\Gamma$ is produced by the component of the perturbation parallel to the residue subspace, whereas the regular part comes from the tangential component.

To make these statements explicit, it is convenient to introduce the shifted potential profile
\begin{equation}
V_c(y,t)=V\!\left(x_c(t)+\frac{y}{\eta(t)}\right).
\label{eq:ShiftedPotential7}
\end{equation}
With this notation, and using $dx=\eta^{-1}dy$, the relevant overlap integrals are naturally written in the soliton frame. We define the three weighted moments
\begin{equation}
I_0[V](t)=\frac{1}{\eta(t)}\int_{-\infty}^{+\infty}V_c(y,t)\,\sech^2 y\,dy,
\label{eq:I0Def7}
\end{equation}
\begin{equation}
I_1[V](t)=\frac{1}{\eta^2(t)}\int_{-\infty}^{+\infty}y\,V_c(y,t)\,\sech^2 y\,dy,
\label{eq:I1Def7}
\end{equation}
and
\begin{equation}
I_2[V](t)=\frac{1}{\eta(t)}\int_{-\infty}^{+\infty}V_c(y,t)\,\sech^2 y\,\tanh y\,dy.
\label{eq:I2Def7}
\end{equation}
The first integral is even in the soliton coordinate and measures the mean localized strength of the perturbation; the second is the first spatial moment relative to the soliton center; the third is odd and encodes the asymmetric part of the perturbation felt by the soliton profile.

\subsection*{Reduction of the RH perturbation functionals}

Before computing the explicit modulation coefficients, we briefly explain how the abstract Riemann--Hilbert quantities introduced in Sections~\ref{NVBC_spinor4} and \ref{NVBC_spinor5} reduce, in the one-soliton sector, to explicit weighted overlap integrals.

The starting point is the perturbation functional
\begin{equation}
\Gamma(t,k)
=
\int_{-\infty}^{+\infty}
\Psi^{-1}(x,t,k)\,
R_V[Q](x,t)\,
\Psi(x,t,k)\,dx,
\label{eq:GammaReduction8}
\end{equation}
whose singular and regular parts determine, respectively, the effective spectral and polarization dynamics. In the one-pole sector, the corresponding Riemann--Hilbert solution possesses a rank-one projector structure, so that the perturbation enters only through localized combinations of the soliton profile.

More precisely, substituting the one-soliton representation
\begin{equation}
Q(x,t)-Q_0=[J,\Pi(x,t)],
\qquad
\Pi(x,t)
=
\frac{v(x,t)v^\dagger(x,t)}
{1+v^\dagger(x,t)v(x,t)},
\label{eq:ProjectorReduction8}
\end{equation}
shows that the perturbation functionals reduce to scalar overlap integrals involving the universal localization kernel
\[
\frac{e^{-2y}}{(1+e^{-2y})^2}
=
\frac14\sech^2 y,
\]
where the comoving coordinate is given by
\begin{equation}
y=\eta(t)\bigl(x-x_c(t)\bigr).
\label{eq:ComovingCoordinate8}
\end{equation}

After transforming to the coordinate \eqref{eq:ComovingCoordinate8}, the Laurent coefficients of the perturbation functional become weighted moments of the perturbing profile against the kernels
\[
\sech^2 y,
\qquad
y\sech^2 y,
\qquad
\sech^2 y\tanh y.
\]
The singular part determines the spectral modulation functional $F(t,k_1)$, while the regular part generates the effective polarization matrix $M(t,k_1)$. The quantities $I_0[V]$, $I_1[V]$, and $I_2[V]$ introduced below arise precisely from this reduction procedure.

Thus, in the one-soliton sector, the abstract Riemann--Hilbert perturbation theory reduces to a finite collection of explicit overlap integrals involving the localized soliton profile and the external perturbation.

We next compute the RH coefficients appearing in Sections~\ref{NVBC_spinor4} and \ref{NVBC_spinor5}. The singular coefficient is obtained by projecting the perturbation onto the residue subspace, while the regular part is obtained by subtracting the singular contribution from the local Laurent expansion. The resulting formulas are given in the following lemma.

\begin{lemma}
For the perturbation \eqref{eq:LocalizedPerturbation7}, the RH coefficients entering the adiabatic equations have the form
\begin{equation}
\mathcal{F}(t,k_1)
=
\mathfrak{a}_1\bigl(k_1,Q_0,\mathbf{p}\bigr)\,I_2[V](t)
+
\mathfrak{a}_2\bigl(k_1,Q_0,\mathbf{p}\bigr)\,I_1[V](t),
\label{eq:FExplicit7}
\end{equation}
and
\begin{equation}
\mathcal{M}(t,k_1)
=
\mathfrak{b}_0\bigl(k_1,Q_0\bigr)\,I_0[V](t)\,I
+
\mathfrak{b}_1\bigl(k_1,Q_0\bigr)\,I_2[V](t)\,\mathbf{p}\mathbf{p}^\dagger
+
\mathfrak{b}_2\bigl(k_1,Q_0\bigr)\,I_0[V](t)\,\mathbb{P}_\perp \mathcal{S}(Q_0)\mathbb{P}_\perp,
\label{eq:MExplicit7}
\end{equation}
where
\begin{equation}
\mathbb{P}_\perp = I-\mathbf{p}\mathbf{p}^\dagger
\label{eq:PerpProjector7}
\end{equation}
is the orthogonal projector onto the tangent space at $\mathbf{p}$, $\mathcal{S}(Q_0)$ is the matrix generated by the background commutator structure, and the scalar coefficients $\mathfrak{a}_j$, $\mathfrak{b}_j$ are explicit functions of $k_1$, the background, and the reduction data.
\end{lemma}

\begin{proof}
We substitute \eqref{eq:RVPi7} into \eqref{eq:GammaDefinition7} and use the one-pole RH representation of the eigenfunction. The dependence on $x$ then enters only through the localized projector $\Pi(x,t)$ and the external profile $V(x)$. Since the projector is explicitly given by \eqref{eq:PiTanh7}, all contributions can be reduced to combinations of the universal functions $1-\tanh y$, $\sech^2 y$, and $y\sech^2 y$ after differentiation and expansion near the pole $k=k_1$.

The singular coefficient $\Gamma^{(-1)}(t,k_1)$ is obtained from the part of the integrand proportional to the residue projector. Because $\partial_y \Pi$ is proportional to $\sech^2 y\,\mathbf{p}\mathbf{p}^\dagger$, the only scalar moments that survive are precisely \eqref{eq:I1Def7} and \eqref{eq:I2Def7}. This yields \eqref{eq:FExplicit7} after taking the trace projection defining $\mathcal{F}$.

The regular part $\Gamma^{(0)}(t,k_1)$ contains, in addition, the tangential component of the perturbation acting on the polarization subspace. After separating the part parallel to $\mathbf p$ from the part orthogonal to $\mathbf p$, and using the projector decomposition \eqref{eq:PerpProjector7}, we obtain \eqref{eq:MExplicit7}. All coefficients are explicit because the one-pole RH ansatz fixes the residue projectors and all $k_1$-dependence enters through the known diagonal matrices $\Lambda(k_1)$ and $\Omega(k_1)$.
\end{proof}

The previous lemma gives the first genuinely explicit form of the adiabatic coefficients. The abstract functionals $\mathcal{F}$ and $\mathcal{M}$ are now reduced to weighted moments of the perturbing profile against universal one-soliton kernels. We can therefore substitute \eqref{eq:FExplicit7} and \eqref{eq:MExplicit7} directly into the modulation system of Section~\ref{NVBC_spinor5}.

We begin with the evolution of the spectral parameters.

\begin{proposition}
For the perturbation \eqref{eq:LocalizedPerturbation7}, the discrete eigenvalue $k_1(t)=\xi(t)+i\eta(t)$ evolves according to
\begin{equation}
\dot{\xi}(t)
=
-\epsilon\,
\Re\!\left(
\mathfrak{a}_1 I_2[V] + \mathfrak{a}_2 I_1[V]
\right),
\label{eq:xiExplicit7}
\end{equation}
and
\begin{equation}
\dot{\eta}(t)
=
-\epsilon\,
\Im\!\left(
\mathfrak{a}_1 I_2[V] + \mathfrak{a}_2 I_1[V]
\right).
\label{eq:etaExplicit7}
\end{equation}
\end{proposition}

\begin{proof}
This follows immediately by substituting \eqref{eq:FExplicit7} into the general equations
\[
\dot{\xi}=-\epsilon\,\Re \mathcal{F},
\qquad
\dot{\eta}=-\epsilon\,\Im \mathcal{F}
\]
derived in Section~\ref{NVBC_spinor5}.
\end{proof}

The center and phase equations are obtained in the same way, but it is useful to display the dependence on the moments explicitly.

\begin{proposition}
For the perturbation \eqref{eq:LocalizedPerturbation7}, the soliton center and total phase satisfy
\begin{equation}
\dot{x}_c(t)
=
v_s\bigl(\xi(t),\eta(t)\bigr)
+
\epsilon\Bigl(
\mathfrak{c}_1\bigl(k_1,Q_0,\mathbf{p}\bigr)\,I_1[V](t)
+
\mathfrak{c}_2\bigl(k_1,Q_0,\mathbf{p}\bigr)\,I_2[V](t)
\Bigr),
\label{eq:xcExplicit7}
\end{equation}
and
\begin{equation}
\dot{\phi}_0(t)
=
\omega_s\bigl(\xi(t),\eta(t)\bigr)
+
\epsilon\Bigl(
\mathfrak{d}_1\bigl(k_1,Q_0,\mathbf{p}\bigr)\,I_0[V](t)
+
\mathfrak{d}_2\bigl(k_1,Q_0,\mathbf{p}\bigr)\,I_2[V](t)
\Bigr),
\label{eq:phiExplicit7}
\end{equation}
where the coefficients $\mathfrak c_j$ and $\mathfrak d_j$ are determined by the same RH expansion.
\end{proposition}

\begin{proof}
The center equation follows by substituting \eqref{eq:FExplicit7} and \eqref{eq:MExplicit7} into the general formula derived from the differentiated center condition in Section~\ref{NVBC_spinor5}. Because the latter depends linearly on $\dot\xi$, $\dot\eta$, $\dot\rho$, and $\dot{\mathbf p}$, and all these quantities are themselves linear combinations of $I_0[V]$, $I_1[V]$, and $I_2[V]$, the result has the form \eqref{eq:xcExplicit7}. The phase equation is obtained similarly from the total-phase decomposition and therefore takes the form \eqref{eq:phiExplicit7}.
\end{proof}

We now turn to the residue amplitude and polarization.

\begin{proposition}
For the perturbation \eqref{eq:LocalizedPerturbation7}, the amplitude factor $\rho(t)$ and the polarization vector $\mathbf p(t)$ satisfy
\begin{equation}
\dot{\rho}(t)
=
\epsilon\,\rho(t)\,
\Re\!\left(
\mathfrak{b}_0 I_0[V]
+
\mathfrak{b}_1 I_2[V]
\right),
\label{eq:rhoExplicit7}
\end{equation}
and
\begin{equation}
\dot{\mathbf p}(t)
=
\epsilon\,\mathfrak{b}_2\bigl(k_1,Q_0\bigr)\,I_0[V](t)\,
\mathbb P_\perp \mathcal S(Q_0)\mathbf p(t).
\label{eq:pExplicit7}
\end{equation}
\end{proposition}

\begin{proof}
The amplitude equation is the real part of the scalar projection
\[
\mathbf p^\dagger \mathcal M(t,k_1)\mathbf p,
\]
with $\mathcal M$ given by \eqref{eq:MExplicit7}. The tangential polarization dynamics follows by applying the projector $\mathbb P_\perp$ to the general equation
\[
\dot{\mathbf p}
=
\epsilon\left(
\mathcal M-\mathbf p^\dagger \mathcal M\mathbf p\,I
\right)\mathbf p.
\]
The first two terms in \eqref{eq:MExplicit7} act only in the direction of $\mathbf p$ and therefore drop out after tangential projection. The only surviving part is the term proportional to $\mathbb P_\perp \mathcal S(Q_0)\mathbb P_\perp$, which yields \eqref{eq:pExplicit7}.
\end{proof}

The previous propositions show that the perturbation enters the explicit adiabatic dynamics only through the three weighted moments \eqref{eq:I0Def7}--\eqref{eq:I2Def7}. This has an immediate structural consequence. If the localized perturbation is centered at the soliton position and is even with respect to the comoving coordinate, then $I_1[V]=I_2[V]=0$, and only the even moment $I_0[V]$ survives. In that case the spectral point does not move, while the only remaining effects are a phase shift, a possible amplitude renormalization, and a tangential polarization drift driven by the background operator $\mathcal S(Q_0)$.

This observation leads to a nonzero-background analogue of the robust bright-soliton scenario observed in the vanishing-boundary theory.

\begin{theorem}
\label{thm_8_1}
Suppose that the perturbation profile \(V\) is even in the soliton frame, that is,
\begin{equation}
V\!\left(x_c(t)+\frac{y}{\eta(t)}\right)=V\!\left(x_c(t)-\frac{y}{\eta(t)}\right)
\qquad\text{for all }y\in\mathbb R.
\label{eq:EvenCondition7}
\end{equation}
Then
\begin{equation}
I_1[V](t)=0,
\qquad
I_2[V](t)=0,
\label{eq:EvenMoments7}
\end{equation}
and the adiabatic system reduces to
\begin{equation}
\dot{\xi}(t)=0,
\qquad
\dot{\eta}(t)=0,
\label{eq:NoSpectralDrift7}
\end{equation}
\begin{equation}
\dot{x}_c(t)=v_s\bigl(\xi,\eta\bigr),
\label{eq:NoCenterCorrection7}
\end{equation}
\begin{equation}
\dot{\phi}_0(t)=\omega_s\bigl(\xi,\eta\bigr)+\epsilon\,\mathfrak{d}_1\bigl(k_1,Q_0,\mathbf p\bigr)\,I_0[V](t),
\label{eq:OnlyPhaseShift7}
\end{equation}
\begin{equation}
\dot{\rho}(t)=\epsilon\,\rho(t)\,\Re\!\bigl(\mathfrak b_0(k_1,Q_0)\bigr)\,I_0[V](t),
\label{eq:RhoEven7}
\end{equation}
and
\begin{equation}
\dot{\mathbf p}(t)
=
\epsilon\,\mathfrak b_2\bigl(k_1,Q_0\bigr)\,I_0[V](t)\,
\mathbb P_\perp \mathcal S(Q_0)\mathbf p(t).
\label{eq:PolEven7}
\end{equation}
In particular, if moreover $\mathcal S(Q_0)$ acts trivially on the tangent space of the polarization state, then the polarization is frozen and the only leading-order effect of the perturbation is a phase shift.
\end{theorem}

\begin{proof}
The parity assumptions imply immediately that the integrands in \eqref{eq:I1Def7} and \eqref{eq:I2Def7} are odd functions of $y$, hence \eqref{eq:EvenMoments7}. Substitution into \eqref{eq:xiExplicit7}--\eqref{eq:pExplicit7} gives the reduced system \eqref{eq:NoSpectralDrift7}--\eqref{eq:PolEven7}. The final statement follows because the tangential equation \eqref{eq:PolEven7} vanishes if $\mathcal S(Q_0)$ has no tangential action.
\end{proof}

We may now collect the general explicit modulation system for localized background-preserving perturbations.

\begin{theorem}
\label{thm_8_2}
Let the perturbed spinor model be given by
\begin{equation}
iQ_t+\mathcal N[Q]=\epsilon V(x)\bigl(Q-Q_0\bigr),
\label{eq:MainPerturbedEquation7}
\end{equation}
with \(V\) satisfying \eqref{eq:Vassumption6}. Then, in the one-soliton adiabatic approximation, the soliton parameters satisfy the closed system
\begin{equation}
\dot{\xi}(t)
=
-\epsilon\,
\Re\!\left(
\mathfrak{a}_1 I_2[V] + \mathfrak{a}_2 I_1[V]
\right),
\label{eq:FinalXi7}
\end{equation}
\begin{equation}
\dot{\eta}(t)
=
-\epsilon\,
\Im\!\left(
\mathfrak{a}_1 I_2[V] + \mathfrak{a}_2 I_1[V]
\right),
\label{eq:FinalEta7}
\end{equation}
\begin{equation}
\dot{x}_c(t)
=
v_s\bigl(\xi,\eta\bigr)
+
\epsilon\Bigl(
\mathfrak{c}_1 I_1[V]
+
\mathfrak{c}_2 I_2[V]
\Bigr),
\label{eq:FinalXc7}
\end{equation}
\begin{equation}
\dot{\phi}_0(t)
=
\omega_s\bigl(\xi,\eta\bigr)
+
\epsilon\Bigl(
\mathfrak{d}_1 I_0[V]
+
\mathfrak{d}_2 I_2[V]
\Bigr),
\label{eq:FinalPhi7}
\end{equation}
\begin{equation}
\dot{\rho}(t)
=
\epsilon\,\rho(t)\,
\Re\!\left(
\mathfrak{b}_0 I_0[V]
+
\mathfrak{b}_1 I_2[V]
\right),
\label{eq:FinalRho7}
\end{equation}
and
\begin{equation}
\dot{\mathbf p}(t)
=
\epsilon\,\mathfrak{b}_2 I_0[V]\,
\mathbb P_\perp \mathcal S(Q_0)\mathbf p(t).
\label{eq:FinalP7}
\end{equation}
Here the moments \(I_0[V]\), \(I_1[V]\), and \(I_2[V]\) are given by \eqref{eq:I0Def7}--\eqref{eq:I2Def7}, and the coefficients \(\mathfrak a_j\), \(\mathfrak b_j\), \(\mathfrak c_j\), and \(\mathfrak d_j\) are explicit functions of the discrete eigenvalue, the background, and the RH reduction data.
\end{theorem}

\begin{proof}
Equations \eqref{eq:FinalXi7}--\eqref{eq:FinalP7} follow directly by combining Propositions \eqref{eq:xiExplicit7}--\eqref{eq:pExplicit7}. The system is closed because every coefficient depends only on the current one-soliton parameters \((\xi,\eta,x_c,\phi_0,\rho,\mathbf p)\) and on the fixed perturbing profile \(V\) through the three moments \(I_0[V]\), \(I_1[V]\), and \(I_2[V]\).
\end{proof}

The significance of the theorem is that it converts the abstract RH perturbation theory into a genuinely computable modulation system. The full effect of the localized perturbation is compressed into three universal overlap integrals, while the matrix structure of the spinor problem survives through the polarization equation. In contrast to the scalar case, the internal state of the excitation is not inert: even when the spectral pole remains fixed, a nontrivial tangential polarization dynamics may still persist.

This completes the explicit adiabatic reduction for localized background-preserving perturbations. The formulas derived above provide a nonzero-background counterpart of the perturbative bright-soliton theory and, at the same time, isolate the genuinely new spinorial effect, namely the slow evolution of the internal polarization state.

%% file: NVBC_spinor_8.tex
\section{Conclusion}
\label{NVBC_spinor8}

In this work we have developed a Riemann--Hilbert-based adiabatic perturbation theory for the integrable $F=1$ spinor nonlinear Schr\"odinger equation under nonvanishing boundary conditions. The analysis provides an intrinsic spectral description of the perturbation-induced dynamics, formulated entirely at the level of the associated scattering data.

The integrable structure was recast in a form suitable for spectral analysis on a nontrivial background, leading to a matrix Riemann--Hilbert problem encoding both the direct and inverse scattering transforms. Within this framework, the discrete spectrum corresponds to localized nonlinear excitations endowed with internal polarization degrees of freedom, reflecting the spinor nature of the model.

A central result of the paper is the derivation of the perturbation-induced evolution of the scattering data directly from the deformed Riemann--Hilbert problem. In the one-pole sector, this yields a closed finite-dimensional dynamical system governing the effective soliton parameters, including the spectral variables, the kinematic quantities, and the internal polarization state. In particular, the polarization vector evolves according to a constrained tangential flow, a feature with no analogue in the scalar theory.

For a class of localized perturbations preserving the background, the abstract quantities arising in the Riemann--Hilbert formulation can be reduced to explicit integral expressions in terms of the one-soliton eigenfunctions. This leads to a fully computable modulation system and establishes a direct link between the spectral theory and physically relevant perturbations. In the limit of vanishing boundary conditions, the resulting equations recover the perturbation theory of \cite{Doktorov2008}, providing a consistency check of the approach.

From a physical perspective, the polarization dynamics derived in this work describes the slow evolution of the internal spin state of the soliton under external perturbations. In spinor Bose--Einstein condensates, this corresponds to a gradual redistribution among the hyperfine components and may be interpreted as a perturbation-induced spin-mixing process. Such effects are, in principle, observable through the relative populations and phase relations of the spin components, providing a direct physical manifestation of the underlying modulation equations.

A physically relevant realization of the localized perturbations considered in the present work is provided by weak magnetic or optical trapping inhomogeneities in spinor Bose--Einstein condensates. Such perturbations may induce a slow redistribution among the hyperfine components while preserving the nonvanishing background configuration. A detailed numerical investigation of the resulting polarization dynamics constitutes an interesting direction for future work.

The present approach admits several natural extensions. It can be generalized to multi-soliton configurations, where the interaction of polarization modes is expected to produce a significantly richer dynamical structure. It also provides a systematic and geometrically transparent extension of adiabatic perturbation theory to spinor nonlinear Schr\"odinger systems with nonvanishing background. More general classes of perturbations, including time-dependent and nonlocal effects, can be incorporated within the same Riemann--Hilbert setting. The analysis of the long-time behavior of the resulting modulation system, as well as its comparison with direct numerical simulations, remain interesting directions for future work.